\documentclass[a4paper,UKenglish,cleveref,autoref,thm-restate]{lipics-v2021}

\hideLIPIcs  %uncomment to remove references to LIPIcs series (logo, DOI, ...), e.g. when preparing a pre-final version to be uploaded to arXiv or another public repository
 
\hypersetup{
    colorlinks=true,
    linkcolor=black,
    citecolor=black,
    urlcolor=blue
    }

\graphicspath{{./figures/}}%helpful if your graphic files are in another directory

\title{Linear Layouts of Graphs with Priority Queues}

\author{Emilio {Di Giacomo}}{University of Perugia, Italy \and \url{https://mozart.diei.unipg.it/digiacomo/}}{emilio.digiacomo@unipg.it}{https://orcid.org/0000-0002-9794-1928}{}%TODO mandatory, please use full name; only 1 author per \author macro; first two parameters are mandatory, other parameters can be empty. Please provide at least the name of the affiliation and the country. The full address is optional. Use additional curly braces to indicate the correct name splitting when the last name consists of multiple name parts.

\author{Walter {Didimo}}{University of Perugia, Italy \and \url{https://mozart.diei.unipg.it/didimo/}}{walter.didimo@unipg.it}{https://orcid.org/0000-0002-4379-6059}{}

\author{Henry {Förster}}{Technical University of Munich, Heilbronn, Germany \and \url{https://www.cs.cit.tum.de/algo/staff/henry-foerster/}}{henry.foerster@tum.de}{https://orcid.org/0000-0002-1441-4189}{}

\author{Torsten Ueckerdt}{Karlsruhe Institute of Technology, Germany \and \url{https://i11www.iti.kit.edu/members/torsten_ueckerdt/index}}{torsten.ueckerdt@kit.edu}{https://orcid.org/0000-0002-0645-9715}{}

\author{Johannes {Zink}}{Technical University of Munich, Heilbronn, Germany \and \url{https://www.cs.cit.tum.de/algo/staff/johannes-zink/}}{johannes.zink@tum.de}{https://orcid.org/0000-0002-7398-718X}{}

\authorrunning{E. Di Giacomo, W. Didimo, H. Förster, T. Ueckerdt, J. Zink} %TODO mandatory. First: Use abbreviated first/middle names. Second (only in severe cases): Use first author plus 'et al.' mandatory. First: Use abbreviated first/middle names. Second (only in severe cases): Use first author plus 'et al.'

\Copyright{Emilio Di Giacomo, Walter Didimo, Henry Förster, Torsten Ueckerdt, and Johannes Zink} %TODO mandatory, please use full first names. LIPIcs license is "CC-BY";  http://creativecommons.org/licenses/by/3.0/

 %TODO mandatory; please add comma-separated list of keywords

\category{} %optional, e.g. invited paper

\relatedversion{} %optional, e.g. full version hosted on arXiv, HAL, or other respository/website

\funding{Research of E.\ Di Giacomo and W.\ Didimo partially supported by: (i) Italian MUR PRIN Proj.\ 2022TS4Y3N -- ``EXPAND: scalable algorithms for EXPloratory Analyses of heterogeneous and dynamic Networked Data''; (ii) Italian MUR PRIN Proj.\ 2022ME9Z78 -- ``NextGRAAL: Next-generation algorithms for constrained GRAph visuALization''; (iii) Italian MUR PON Proj.\ ARS01 00540 -- ``RASTA: Realt\`a Aumentata e Story-Telling Automatizzato per la valorizzazione di Beni Culturali ed Itinerari''}%

\acknowledgements{This work was initiated at the Annual Workshop on Graph and Network Visualization (GNV~2023), Chania, Greece, June 2023.}%optional

\nolinenumbers %uncomment to disable line numbering

\usepackage{xspace}
\usepackage[disable]{todonotes}
\usepackage{thmtools} 
\usepackage{thm-restate}
\usepackage{apptools}

\usepackage{complexity}

\newcommand{\restateref}[1]{\IfAppendix{\hyperref[#1]{$\star$}}{\hyperref[#1*]{$\star$}}}

\newcommand{\invref}[2]{\textcolor{lipicsGray}{\sffamily\bfseries\upshape\mathversion{bold}{#1.\ref{#2}}}}
\DeclareMathOperator{\pqn}{pqn}
\DeclareMathOperator{\bpqn}{spqn}

\begin{document}

\maketitle

\begin{abstract}%STACS2025
	A \emph{linear layout} of a graph consists of a linear ordering of its vertices and a partition of its edges into \emph{pages} such that the edges assigned to the same page obey some constraint.
	The two most prominent and widely studied types of linear layouts are \emph{stack} and \emph{queue layouts}, in which any two edges assigned to the same page are forbidden to \emph{cross} and \emph{nest}, respectively.
%    Stack and queue layouts find application in a variety of practical problems, including scheduling, VLSI design, fault-tolerant multiprocessing, RNA folding, and traffic-light control.
	The names of these two layouts derive from the fact that, when parsing the graph according to the linear vertex ordering, the edges in a single page can be stored using a single stack or queue, respectively.
%	Edges are inserted into the data structure when encountering their first endpoint and removed from the data structure when encountering their second endpoint.
	Recently, the concepts of stack and queue layouts have been extended by using a double-ended queue or a restricted-input queue for storing the edges of a page.
	We extend this line of study to \emph{edge-weighted} graphs by introducing \emph{priority queue layouts}, that is, the edges on each page are stored in a priority queue whose keys are the edge weights.
	First, we show that there are edge-weighted graphs that require a linear number of priority queues.
    Second, we characterize the graphs that admit a priority queue layout with a single queue, regardless of the edge-weight function, and we provide an efficient recognition algorithm.
    Third, we show that the number of priority queues required independently of the edge-weight function is bounded by the pathwidth of the graph,
    but can be arbitrarily large already for graphs of treewidth two.
	Finally, we prove that determining the minimum number of priority queues is \NP-complete if the linear ordering of the vertices is fixed.
\end{abstract}

%%%%%%%%%%
\section{Introduction}\label{sec:intro}

% WALTER: We may want to mention the following application for PQ-layouts somewhere.
%The edges of the graph represent jobs, each with a given priority. The vertices represent time constraints between jobs; namely, two adjacent edges represent jobs that
%either have start at the same time, or have to end at the same time, or one has to start at the same time the other ends. If any two jobs are active and have different priorities, the one with smaller priority must end before the other (or they end at the same time).
%A scheduling of the jobs that respects these constraints is feasible if and only if there exists a PQ-layout of the graph on a single page.

\emph{Stack layouts} were introduced by Bernhart and Kainen in 1979~\cite{DBLP:journals/jct/BernhartK79} with the motivation to study a ``quite natural'' edge decomposition technique.
A stack layout is a \emph{linear layout}, that is, the vertices are positioned on a line $\ell$ with the edges being embedded on several half-planes, called \emph{pages}, delimited by $\ell$.
In a stack layout, edges embedded on the same page are forbidden to cross,
a rule that contributed to their original designation as \emph{book embeddings}.
An example of stack layout on one page is depicted in \cref{fi:stack-queue-b}. 
The minimum number of pages required to obtain a stack layout of a graph is called the \emph{stack number}.
After their introduction, stack layouts attracted notable interest
in the field of theoretical computer science due to their capacity to model certain aspects in several domains,
including VLSI design~\cite{diogenes,10.1137/0608002}, fault-tolerant multiprocessing \cite{10.1137/0608002,DBLP:journals/tc/Rosenberg83}, RNA folding~\cite{hs-99}, and traffic-light control~\cite{kainen-90}. This led, in particular, to a series of papers investigating upper bounds on the stack number of planar graphs~\cite{DBLP:conf/stoc/BussS84,DBLP:conf/focs/Heath84,Istrail1988a} and culminating at an upper bound of four~\cite{DBLP:journals/jcss/Yannakakis89}. Bekos et al.~\cite{DBLP:journals/jocg/KaufmannBKPRU20} and Yannakakis~\cite{DBLP:journals/jctb/Yannakakis20} independently proved that this upper bound is tight.%, showing that there are planar graphs with stack number four. 
%Stack layouts also find applications in RNA folding~\cite{hs-99}, traffic-light control~\cite{kainen-90}.

The name stack layout derives from the fact that the edges assigned to the same page can be stored using a stack as follows~\cite{diogenes}.
Assume that $\ell$ is a horizontal line and consider the left-to-right ordering $v_1,\ldots,v_n$ of the vertices.
Now, consider a left-to-right sweep of $\ell$ and the set $E_p$ of edges assigned to a page~$p$.
Initially, the stack storing $E_p$ is empty.
At vertex $v_i$, all edges that have $v_i$ as their right endpoint are popped from the stack.
Since all vertices are at the same side of the half-plane delimited by~$\ell$, the edges ending at $v_i$ are necessarily the last ones that have started prior to encountering~$v_i$,
so they are precisely the edges on top of the stack.
Afterwards, edges starting at $v_i$ are pushed onto the stack.
(Note that the order of the pushes is determined by the right endpoints of the corresponding edges.)
\begin{figure}
	\centering
    \begin{subfigure}[t]{0.28\textwidth}
        \centering
        \includegraphics[page=4]{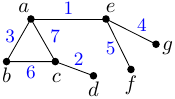}
        \subcaption{Edge-weighted graph $G$.}
        \label{fi:stack-queue-a}
    \end{subfigure}
	\hfill
	\begin{subfigure}[t]{0.23\textwidth}
		\centering
		\includegraphics[page=2]{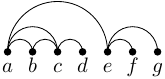}
		\subcaption{Stack layout of $G$.}
		\label{fi:stack-queue-b}
	\end{subfigure}
	\hfill
	\begin{subfigure}[t]{0.23\textwidth}
		\centering
		\includegraphics[page=3]{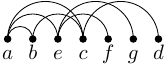}
		\subcaption{Queue layout of $G$.}
		\label{fi:stack-queue-c}
	\end{subfigure}
	\hfill
	\begin{subfigure}[t]{0.23\textwidth}
		\centering
		\includegraphics[page=5]{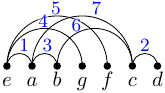}
		\subcaption{PQ-layout of $G$.}
		\label{fi:stack-queue-d}
	\end{subfigure}
    
	\caption{(a)--(c) Examples of stack and queue layouts on one page for the same graph $G$ ignoring edge weights. (d) A PQ-layout on one page of $G$ (edge weights are in blue).}
	\label{fi:stack-queue}
\end{figure}

Based on this latter interpretation of stack layouts, Heath, Leighton and Rosenberg~\cite{DBLP:journals/siamdm/HeathLR92,DBLP:journals/siamcomp/HeathR92} proposed to study \emph{queue layouts} of graphs where the edges occurring on the same page of the linear layout can be represented by a queue.
More precisely, edges ending at $v_i$ are first dequeued from a queue, after which edges beginning at $v_i$ are enqueued during the left-to-right sweep of $\ell$.
Thus, edges on the same page are allowed to cross but forbidden to properly nest. An example of a queue layout on one page is shown in \cref{fi:stack-queue-c}.
Surprisingly, these quite non-planar representations still found many practical applications, including scheduling~\cite{DBLP:journals/jpdc/BhattCLR96}, VLSI~\cite{DBLP:journals/siamcomp/LeightonR86}, and, a decade after their inception, 3D crossing-free graph drawing~\cite{DBLP:journals/comgeo/GiacomoLM05,DBLP:journals/siamcomp/DujmovicMW05}.
Similarly to stack layouts, one is interested in minimizing the number of pages required to obtain a queue layout, which is known as the \emph{queue number}.
%
%In contrast to the stack number, the queue number of planar graphs remained elusive for several decades during which only sublinear (but non-constant) upper bounds have been found~\cite{DBLP:journals/algorithmica/AlamBGKP20,DBLP:journals/algorithmica/BannisterDDEW19,DBLP:journals/siamcomp/BattistaFP13,DBLP:journals/jct/Dujmovic15}.
%%\todo{jz: I don't find ``non-linear'' ideal here because it generally sounds a bit like ``more than linear''. I think ``non-constant (but sublinear)'' might be better(or: ``only sublinear (but no constant)'')}
%Note that the first significant difference between planar and non-planar graphs has been only proposed in 2019 when it was shown that bounded degree planar graphs have bounded queue number~\cite{DBLP:journals/siamcomp/BekosFGMMRU19}, whereas non-planar bounded degree graphs were already known to not have this property~\cite{DBLP:journals/dmtcs/Wood08}.
%
%In the past couple years, two break-through results have answered the two must urgent open problems in the field.
%First, Dujmovi\'{c} et al.~\cite{DBLP:journals/jacm/DujmovicJMMUW20} proved that the queue number of planar graphs is bounded\footnote{Minor improvements have been proposed subsequently~\cite{DBLP:journals/algorithmica/BekosGR23,DBLP:journals/corr/abs-2305-16087}.}.
%Second, it has been shown independently by Bekos et al.~\cite{DBLP:journals/jocg/KaufmannBKPRU20} and Yannakakis~\cite{DBLP:journals/jctb/Yannakakis20} that there are planar graphs with stack number at least four. 
%
However, in contrast to the stack number, the queue number of planar graphs remained elusive for several decades during which only sublinear (but non-constant) upper bounds have been found~\cite{DBLP:journals/algorithmica/AlamBGKP20,DBLP:journals/algorithmica/BannisterDDEW19,DBLP:journals/siamcomp/BattistaFP13,DBLP:journals/jct/Dujmovic15}.
The first significant difference between planar and non-planar graphs was only demonstrated in 2019 when it was shown that bounded degree planar graphs have bounded queue number~\cite{DBLP:journals/siamcomp/BekosFGMMRU19}, whereas non-planar bounded degree graphs were already known to not have this property~\cite{DBLP:journals/dmtcs/Wood08}. Recently, Dujmovi\'{c} et al.~\cite{DBLP:journals/jacm/DujmovicJMMUW20} proved that the queue number of planar graphs is bounded by a constant.
Minor improvements have since been made~\cite{DBLP:journals/algorithmica/BekosGR23,DBLP:journals/corr/abs-2305-16087}.

Another interesting recent direction in investigating linear layouts is to expand beyond stack and queue layouts by using other data structures for partitioning edges.
For instance, \emph{mixed linear layouts} allow the usage of both stack and queue pages~\cite{DBLP:journals/tcs/AngeliniBKM22,DBLP:conf/gd/ColKN19,DBLP:conf/gd/Pupyrev17}, whereas \emph{deque layouts} use \emph{double-ended queues} which allow to insert and pull edges from the head and the tail~\cite{DBLP:conf/gd/AuerBBBG10,DBLP:conf/wg/AuerG11,deque},
and \emph{rique layouts} use \emph{restricted-input queues}
which are double-ended queues where insertions are allowed to occur only at the head~\cite{DBLP:conf/gd/BekosFKKKR22}.
%
%We emphasize that practical applications for these new concepts may emerge in the future similar to how important applications for queue and stack layouts emerged only decades after~\cite{diogenes,DBLP:journals/comgeo/GiacomoLM05,DBLP:journals/siamcomp/DujmovicMW05,10.1137/0608002,DBLP:journals/tc/Rosenberg83} their introduction~\cite{DBLP:journals/jct/BernhartK79,DBLP:journals/siamdm/HeathLR92,DBLP:journals/siamcomp/HeathR92}.

%\todo[inline]{Walter + Emilio (or Torsten / Henry if you want to):
%Reviewer 1:
%``The relation of queue and stack numbers to planarity was broadly discussed. I understand that this had been the most important question in the area for a long time, however, planarity is not the main subject of this paper. Thus, I believe some more space should be devoted to discussing the relations of queue and stack numbers to treewidth, pathwidth, complete graphs, and computational complexity (things that are studied in this paper). Maybe also the result showing that pqn is unbounded for planar graphs should be highlighted more.''}

\subparagraph*{Our contribution.} 
%%%%% PREVIOUS VERSION
%Although stack and queue layouts find applications in a variety of real-world domains, these types of linear layouts (and their subsequent extensions) are intended for graph data without edge weights. Motivated by this observation, our paper extends the notion of linear layouts to edge-weighted graphs, and introduces a model called \emph{priority queue layout}, or simply \emph{PQ-layout}. This model utilizes  \emph{priority queues} for the storage of edges assigned to the same page, where the edge weights correspond to the priorities (i.e., keys) for the data structure.  
%As for stack and queue layouts, in a PQ-layout all vertices of the graph lie on a horizontal line $\ell$, and the edges are assigned to different priority queues according to a left-to-right sweep of $\ell$. When an edge $e$ is dequeued, it must have the minimum priority among the edges stored in the same priority queue as $e$. See \cref{fi:stack-queue-d} for an example, and refer to \cref{sec:basic} for a formal definition of our model. Note that, differently from stack and queue layouts, our model allows edges in the same page to cross or properly nest, provided that they have suitable weights.
%\todo[inline]{Walter: STACS Reviewer 3 complains about the lack of motivation
%for PQ-layouts. Re-visit? Include rebuttal letter from STACS?}
%%%%%%%%%%%%%%%%%%%%%%

%%%%% NEW VERSION
The wide range of applications of linear layouts and the fact that they have been restricted so far to unweighted graphs, naturally motivate extensions of the notion of linear layout to edge-weighted graphs. For example, in scheduling applications, the edges of a graph can represent processes, each with an associated priority; each vertex 
$v$ of the graph enforces that all processes corresponding to edges incident to 
$v$ begin simultaneously; among the running processes, those with higher priority must be completed first.
Edge weights can also model constraints in other classical applications of linear layouts, such as logistics networks modelled as switchyard networks~\cite{DBLP:books/aw/Knuth68,DBLP:conf/stoc/Pratt73}, where adjacent edges represent containers that must be pushed or popped simultaneously across different storage locations.
The weight of an edge corresponds to the weight of its container, and it may be necessary to ensure that containers stacked on top of others are lighter.       

We introduce a model called \emph{priority queue layout}, or simply \emph{PQ-layout}. This model utilizes  \emph{priority queues} for the storage of edges assigned to the same page, where the edge weights correspond to the priorities (i.e., keys) for the data structure.  
As for stack and queue layouts, in a PQ-layout all vertices of the graph lie on a horizontal line $\ell$, and the edges are assigned to different priority queues according to a left-to-right sweep of $\ell$. When an edge $e$ is dequeued, it must have the minimum priority among the edges stored in the same priority queue as $e$. See \cref{fi:stack-queue-d} for an example, and refer to \cref{sec:basic} for a formal definition of our model. Note that, unlike stack and queue layouts, our model allows edges in the same page to cross or properly nest, provided that they have suitable weights.
%%%%%%%%%%%%%%%%%%%%%%%%%%%%%%%%%%%%%%%%%%%%%%%%%%%%%%%%

%and investigates the usage of \emph{priority queues} for the storage of edges assigned to the same page, where the edge weights correspond to the priorities (i.e., keys) for the data structure.  
%%Our paper follows this latter line of research and extends the notion of linear layouts to \emph{edge-weighted} graphs. This naturally motivates to investigate the usage of \emph{priority queues} for the storage of edges assigned to the same page, where the edge weights are to be interpreted as priorities (i.e., keys) for the data structure. 
%We remark that priority queues are common data structures for dealing with data associated with keys, whereas the only previous study on linear layouts of edge-weighted graphs uses the edge weights to further restrict stack layouts~\cite{DBLP:journals/jgaa/BattistaFPT21}.
 
%
%Our model, called \emph{priority queue layout}, or simply \emph{PQ-layout}, may allow for many edges assigned to the same page provided a suitable edge-weight function. 
%Hence, we will mainly consider how many pages are required for PQ-layouts of a graph with an arbitrary edge-weight function.

Following a current focus of research on linear layouts, we mainly study the \emph{priority queue number} of edge-weighted graphs, that is, the number of pages required by their PQ-layouts:
%can be summarized as follows:

%First, we formally introduce PQ-layouts in \cref{sec:basic}.
%Then, we provide non-trivial lower and upper bounds for the \emph{priority queue number} of complete graphs in \cref{sec:complete}.
%Afterwards, we make some observations regarding the priority queue number of subdivisions and minors of graphs in \cref{sec:pqn1:subdivision}.
%Based on these, after establishing technical lemmas in \cref{sec:pqn1:positive} and forbidden minors in \cref{sec:pqn1:negative}, we characterize the graphs admitting PQ-layouts with one page using arbitrary edge weightings in \cref{sec:pqn1:summary} and derive an efficient recognition algorithm.
%Furthermore, we investigate graphs of bounded pathwidth and treewidth in \cref{sec:results-kpq} and show that there are weighted planar graphs with unbounded priority queue number.
%Going back to graphs with fixed edge-weight functions,
%we show in \cref{sec:complexity} that determining the number of required priority queues is \NP-complete even if the order of the vertices is given.
%We conclude the paper with some open problems in \cref{sec:conclusions}.

%\begin{enumerate}

	(i)~We provide non-trivial upper and lower bounds for the priority queue number of complete graphs (\cref{sec:complete}).
%    In particular, we show a linear lower bound in the number of vertices.
    %that there are $n$-vertex edge-weighted graphs with $\Omega(n)$ priority queue number. 	
	(ii)~We characterize the graphs that have priority queue number~1, regardless of the edge-weight function (\cref{sec:pqn1}).
%    The characterization is in terms of forbidden minors, and it is used to devise a linear-time testing and construction algorithm.
	(iii)~We investigate the priority queue number of graphs with bounded pathwidth and bounded treewidth (\cref{sec:results-kpq}).
%    We show that graphs with bounded pathwidth have bounded priority queue number, whereas the priority queue number can be arbitrarily large already for graphs of treewidth two.
	(iv)~We prove that deciding whether an edge-weighted graph has priority queue number~$k$ (for a non-fixed integer $k$) is \NP-complete if the linear ordering of the vertices is fixed (\cref{sec:complexity}).
%\end{enumerate}

We remark that, as far as we know, there is only one previous study in the literature about using linear layouts of edge-weighted graphs~\cite{DBLP:journals/jgaa/BattistaFPT21}; however, this study uses edge weights to further restrict stack layouts.
For space restrictions, proofs of statements marked with a (clickable) ($\star$) are omitted or sketched.
See the appendix for a full version of these proofs.

%\todo[inline]{If possible, find better motivations for the problem (both on the theoretical and on the practical sides).}

%\todo[inline]{In the paper do not use $w$ for denoting vertices.}

%%%%%%%%%%
\section{PQ-Layouts}\label{sec:basic}

%\subparagraph*{Definitions.} 
Let $\langle G,w \rangle$ be an \emph{edge-weighted} graph, that is, $G=(V,E)$ is an undirected graph and $w \colon E \to \mathbb{R}$ is an \emph{edge-weight function}.
A \emph{priority queue layout}, or simply a \emph{PQ-layout}, $\Gamma$ of $\langle G,w \rangle$ consists of a linear (left-to-right) ordering $v_1 \prec \cdots \prec v_n$ of the vertices of $G$ and a partitioning of the edges of $G$ into $k$ sets $\mathcal{P}_1,\ldots,\mathcal{P}_k$,
where each set is called a \emph{page} and is associated with a priority queue.
For each page $\mathcal{P}_i$ with $i \in \{1, \dots, k\}$, the following must hold:
when traversing the linear ordering~$\prec$ (from left to right) and performing, at each encountered vertex $v \in V$, the operations O.\ref{item:o1} and~O.\ref{item:o2} (in this order),
condition C.\ref{item:c1} must always hold.
\begin{enumerate}[{O.}1]
\item \label{item:o1}
Each edge $e \in \mathcal{P}_i$ having $v$ as its right endpoint is pulled from the priority queue.
We do not count $e$ as active any more.
\item \label{item:o2}
Each edge $e \in \mathcal{P}_i$ having $v$ as its left endpoint is inserted into the priority queue of~$\mathcal{P}_i$ with key $w(e)$.
We say that~$e$ becomes \emph{active}.
Two edges being active at the same time (even though they may have become active in different steps) are called \emph{co-active}.
%Change item style from O.2 to C.1 (i.e., decrease the counter by 2 and change the prefix to "C.")
\addtocounter{enumi}{-2}
\renewcommand{\labelenumi}{C.\arabic{enumi}}
\item \label{item:c1}
There is no pair of co-active edges~$e$ and~$e'$
such that $e$ has $v$ as its right endpoint, $e'$ has a distinct right endpoint, and $w(e) > w(e')$.
In other words, among the active edges %in the priority queue
of~$\mathcal{P}_i$,
the edges having $v$ as their right endpoint have the smallest edge weights.
\end{enumerate}
%For a fixed vertex ordering $\prec$, let us also define two edges $e = xy$ and $e' = x'y'$ to be \emph{co-active at vertex $v$} if $x,x' \prec v \preceq y,y'$.
%In other words, when $v$ is visited, then $e$ and $e'$ are both active (though possibly in different priority queues).
%
%\smallskip
%
%PQ-layouts can be equivalently defined by the absence of forbidden configurations.
%If $e=xy$ and $e'=x'y'$ are any two edges with $w(e) > w(e')$, we define (refer to \cref{fi:forbidden-configurations}):
%\begin{itemize}
%\item \textsf{Forbidden Nesting:} $x' \prec x \prec y \prec y'$, i.e., $e$ and $e'$ are vertex-disjoint edges such that the heavier edge is nested inside the lighter edge;
%\item \textsf{Forbidden Pseudo-Nesting:} $x' = x \prec y \prec y'$, i.e., $e$ and $e'$ share their left endpoint such that the heavier edge is nested inside the lighter edge;
%\item \textsf{Forbidden Crossing:} $x \prec x' \prec y \prec y'$, i.e., $e$ and $e'$ are vertex-disjoint edges crossing each other such that lighter edge ends to the right of the heavier edge.
%\end{itemize} 
%Then $\Gamma$ is a PQ-layout of $\langle G,w \rangle$ if and only if no two edges in the same priority queue form any of the forbidden configurations above.

PQ-layouts can be equivalently defined by the absence of the forbidden configuration F.\ref{item:f1}.
\begin{enumerate}[{F.}1]
	\item \label{item:f1}
	For some $i \in \{1, \dots, k\}$, there are two edges $e = uv$ and $e' = u'v'$
	such that $e, e' \in \mathcal{P}_i$, $w(e) > w(e')$, and $u, u' \prec v \prec v'$.
\end{enumerate}
The order of $u$ and $u'$ is irrelevant for F.\ref{item:f1}.
For completeness, the resulting three arrangements of the involved vertices
are illustrated in \cref{fi:forbidden-a,fi:forbidden-b,fi:forbidden-c}.
Note that the forbidden pseudo-nesting in \cref{fi:forbidden-b} requires $u = u'$.
In contrast, the symmetric case where $v = v'$ is never forbidden.
This is a direct consequence of the asymmetry between the insert and pull operations of priority queues.
Thus, unlike for stack and queue layouts, the reverse vertex ordering does not necessarily provide a PQ-layout.

% It is immediate to see that the absence of the forbidden configurations above is also a sufficient condition for PQ-layouts.
% Namely, a linear ordering $\prec$ of the vertices of $G$ along with a partitioning of the edges of $G$ into $k$ sets $\mathcal{P}_1,\ldots,\mathcal{P}_k$ corresponds to a PQ-layout of $\langle G,w \rangle$ if and only if no two edges in any $\mathcal{P}_i$ form any of the forbidden configurations above.

\begin{figure}
	\centering
	\begin{subfigure}[t]{0.225\textwidth}
		\centering
		\includegraphics[page=1]{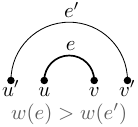}
		\subcaption{Forbidden\\\phantom{\textbf{(a)} }Nesting}
		\label{fi:forbidden-a}
	\end{subfigure}
	\hfil
	\begin{subfigure}[t]{0.225\textwidth}
		\centering
		\includegraphics[page=2]{forbidden}
		\subcaption{Forbidden\\\phantom{\textbf{(b)} }Pseudo-Nesting}
		\label{fi:forbidden-b}
	\end{subfigure}
	\hfil
	\begin{subfigure}[t]{0.225\textwidth}
		\centering
		\includegraphics[page=3]{forbidden}
		\subcaption{Forbidden\\\phantom{\textbf{(c)} }Crossing}
		\label{fi:forbidden-c}
	\end{subfigure}
	\hfil
	\begin{subfigure}[t]{0.225\textwidth}
		\centering
		\includegraphics[page=4]{forbidden}
		\subcaption{$k$-Inversion}
		\label{fi:forbidden-d}
	\end{subfigure}
	\caption{(a)--(c)~Arrangements in a PQ-layout
		that correspond to the Forbidden Configuration~F.\ref{item:f1}. (d)~A $k$-inversion, a forbidden arrangement in linear layouts with priority queue number $k-1$.}
	\label{fi:forbidden-configurations}
\end{figure}

We refer to the linear ordering $\prec$ of vertices also by the ordering along the \emph{spine}.
% For the sake of simplicity, in the following each set $\mathcal{P}_i$ of a PQ-layout will be identified with its associated priority queue.
For a PQ-layout $\Gamma$ of any graph, let $\pqn(\Gamma)$ denote the number of priority queues in $\Gamma$.
Now, the \emph{priority queue number} $\pqn(G,w)$ of $\langle G,w \rangle$ is the minimum integer $k \geq 0$ for which there exists a PQ-layout $\Gamma$ of $\langle G,w \rangle$ with $\pqn(\Gamma) = k$.  
%Moreover, we say that $\langle G,w \rangle$ has \emph{priority queue number} $\pqn(G,w)$ if and only if there is a PQ-layout $\Gamma$ of $\langle G,w \rangle$ with $\pqn(\Gamma)=\pqn(G,w)$ but no PQ-layout $\Gamma'$ of $\langle G,w \rangle$ with $\pqn(\Gamma')<\pqn(G,w)$.

For any given graph $G$, there always exists an edge-weight function $w \colon E \to \mathbb{R}$ such that $\pqn(G,w) = 1$.
It is enough to assign the same weights to all edges but it is also possible to use distinct edge weights.
Namely, define any linear ordering of the vertices $v_1 \prec \cdots \prec v_n$ and visit the vertices in this ordering.
When a vertex $v$ with $v \neq v_1$ is visited, assign consecutive (non-negative) integer weights to all edges for which $v$ is the right endpoint such that the weights monotonically increase. %, in such a way that all these integers are greater than those used when we visited $u$, for any $u \prec v$.
Hence, any graph with any given linear ordering of its vertices can have priority queue number~1 if we can choose the edge-weight function.
However, the edge-weight function may be provided as part of the input.
Thus, it is interesting to study bounds on the priority queue number of families of graphs that hold independently of the edge-weight function.
To this end, for a graph $G=(V,E)$ without any prescribed edge-weight function, the \emph{priority queue number} $\pqn(G)$ is the minimum integer $k \geq 0$ such that for every possible edge-weight function $w\colon E \to \mathbb{R}$ we have $\pqn(G,w) \leq k$.
% there is a PQ-layout $\Gamma$ of $\langle G,w \rangle$ with~$\pqn(\Gamma)=k$.  

% The three forbidden configurations can be further unified as follows. We say that edge $e={v_\ell}v_r$ is \emph{active} on interval $[\ell,r]$ and call $[\ell,r]$ the active interval of $e$. Moreover, we say that two edges $e$ and $e'$ are \emph{co-active} if their intervals have a non-empty intersection. The three forbidden configurations above can be summarized by saying that it is not allowed to have two co-active edges $e={v_\ell}v_r$ and $e'=v_{\ell'}v_{r'}$ with $w(e) > w(e')$ and $r < r'$ (where $\ell < r$ and $\ell' < r'$). Also observe that the specific weights of the weight function have no impact, the important aspect is an ordering of the edges by weight.

% Instead of indices, we also write $u \prec v$
% if a vertex $u$ precedes $v$ in the linear order of vertices.

\smallskip

Some of our proofs utilize %worst-case
configurations where every pair of edges must be assigned to distinct priority queues.
Let $\langle G,w \rangle$ be an edge-weighted graph, let $\prec$ be a vertex ordering of $G$, and let $k$ be an integer such that $k \geq 1$.
We define a \emph{$k$-inversion} to be a sequence $(e_1,\ldots,e_k)$ of $k$ edges in $G$ with pairwise distinct right endpoints $v_1,\ldots,v_k$, respectively, such that (i)~$v_k \prec \cdots \prec v_1$, where $v_k$ and $v_1$ are called the \emph{left} and \emph{right end}, respectively, (ii)~the left endpoint of each of $e_1,\ldots,e_k$ appears to the left of $v_k$ in $\prec$, and (iii)~$w(e_1) < \cdots < w(e_k)$.
See \cref{fi:forbidden-d} for an illustration.
Any two edges in a $k$-inversion form a forbidden configuration and thus must be assigned to distinct priority queues.
In other words, if, for edge weights $w$ of~$G$, every vertex ordering of $\langle G,w \rangle$ contains a $k$-inversion, then $\pqn(G,w) \geq k$ and hence $\pqn(G) \geq k$.

%\subparagraph*{Preliminary Observations.}

\section{Priority Queue Number of Complete Graphs}
\label{sec:complete}

In this section, we concentrate on complete graphs and complete bipartite graphs.
We start with a simple upper bound for the complete graph~$K_n$ on $n$ vertices.

\begin{observation}
$\pqn(K_n) \leq n-1$.
\end{observation}

\begin{proof}
Choose an arbitrary ordering of the vertices $v_1 \prec \cdots \prec v_n$ and define the sets (priority queues) $\mathcal{P}_{1}, \dots, \mathcal{P}_{n-1}$.
For each $j\in \{2,\ldots,n\}$, assign all edges whose right endpoint is~$v_j$ to $\mathcal{P}_{j-1}$.
Clearly, no forbidden configuration occurs, independent of the edge-weight function.
\end{proof}

Next, we consider bipartite graphs, that is, the graphs whose vertex sets are the union of two disjoint independent sets $A$ and $B$.
In the literature on linear layouts,  \emph{separated} linear layouts, where vertices of $A$ precede the vertices of $B$ or vice versa, have been considered for bipartite graphs~\cite{DBLP:journals/dmtcs/DujmovicPW04}.
%In such a layout, the vertices of one part, say $A$, precede the vertices in $B$ in the linear ordering of the vertices.
The \emph{separated priority queue number} $\bpqn(G)$ of a bipartite graph~$G$ is the minimum% number
~$k$ such that, for any edge-weight function, there is a PQ-layout~$\Gamma$ that is a separated linear layout
and $\pqn(\Gamma)=k$.
We can easily determine the separated priority queue number
of the complete bipartite graph~$K_{n,n}$.
%where $|A|=|B|=n$.
At the end of this section, we will use the separated priority queue number
to give a lower bound on the priority queue number of complete and complete bipartite graphs.

\begin{theorem}
\label{thm:bipartitePriorityQueueLayout}
$\bpqn(K_{m,n})=\min\{m, n\}$.
\end{theorem}

\begin{proof}
First observe that $\bpqn(K_{m,n}) \leq \min\{m, n\}$:
place the vertices of the larger set, say $A$, before
the vertices of the smaller set, say $B$.
Then, independent of the edge weight function, for each vertex $b \in B$,
all edges ending at $b$ can be assigned to a separate page.

To show $\bpqn(K_{m,n})\geq \min\{m, n\}$,
assume w.l.o.g.\ that the vertices in $A$ precede the vertices in~$B$.
Consider an edge weight function
where, at every vertex $b \in B$, each edge weight in $\{1, 2, \dots, \min\{m, n\}\}$ occurs at least once among $b$'s incident edges.
In any bipartite PQ-layout~$\Gamma$ of~$K_{m,n}$,
denote the vertices of $B$ by $b_1, b_2, \dots$ in the order they appear along the spine.
In~$\Gamma$, $b_1$ is incident to an edge with weight~$\min\{m, n\}$,
$b_2$ is incident to an edge with weight~$\min\{m, n\}-1$, etc.,
$b_{\min\{m, n\}}$ is incident to an edge of weight~1.
This is a $\min\{m, n\}$-inversion requiring at least~$\min\{m, n\}$ pages.
\end{proof}

\cref{thm:bipartitePriorityQueueLayout} also provides an upper bound for the (non-separated) priority queue number
of~$K_{n,n}$.
Next, we give a corresponding linear lower bound showing that $\pqn(K_{n,n}) \in \Theta(n)$.

\begin{restatable}[\restateref{thm:k_nn}]{theorem}{Knn}
\label{thm:k_nn}
$\pqn(K_{n,n}) \geq \left\lceil \frac{3-\sqrt{5}}{4}n \right\rceil \approx 0.191 n$.
\end{restatable}

\begin{proof}[Proof Sketch]
%\todo[inline]{Emilio: STACS Reviewer 1: ``i wasn't able to follow the idea of a proof here, without looking at appendix; the grid representation alone is nice but doesn't really provide an argument. Consider adding more details here or moving the "proof" to appendix completely.''}
For $K_{n,n}$, we define an edge-weight function~$w$ that has,
at each vertex, for every $i \in \{1,\ldots,n\}$, exactly one incident edge of weight~$i$.\footnote{%
In contrast to the proof of \cref{thm:bipartitePriorityQueueLayout},
we cannot arbitrarily assign weights to edges incident to each vertex $b \in B$
because this may give twice the same edge weight
at a vertex of partition~$A$.}
To this end, we partition the edge set into $n$ perfect matchings $M_1,\ldots,M_n$, and,
for each edge $e \in M_i$, we set $w(e)=i$.

For any PQ-layout~$\Gamma$ of $\langle K_{n,n}, w \rangle$, we can find a sublayout~$\Gamma'$
that is a \textit{bipartite} PQ-layout of~$\langle K_{\frac n2, \frac n2}, w \rangle$~--
either take the first $\frac n2$ vertices of $A$ and the last
$\frac n2$ vertices of $B$ or vice versa.\footnote{%
    For simplicity, we assume that $n$ is even.}

Note that we cannot directly apply \cref{thm:bipartitePriorityQueueLayout}
to $\Gamma'$ because we have already fixed an edge-weight function.
Instead, we analyze the possible distribution of the edge weights
incident to vertices in~$\Gamma'$ by a representation as a black-and-white grid;\footnote{%
    A similar, although not identical, representation is used by Alam et al.~\cite{MDALAM2022131} to prove the mixed page number of $K_{n,n}$.}
see the appendix for an illustration and a more extensive description.
Our grid has $\frac n2$ columns that represent the vertices
of the latter partition (in order from left to right as they appear in~$\Gamma'$)
and it has $n$ rows that represent the edge weights $\{1, \dots, n\}$ (in order from bottom to top).
A grid cell in column~$i$ and row~$j$ is colored black if
the $i$-th vertex is incident to an edge of weight~$j$,
and it is colored white otherwise.
Observe that, in this grid, a \textit{strictly monotonically decreasing path} of black cells
having length~$k$ is equivalent to a $k$-inversion in~$\Gamma'$.
We can show that there always exists such a path of length
$\left\lceil \frac{3 - \sqrt{5}}{4}n \right\rceil$:
the candidates for being the first black cell of the path are those who
do not have another black cell in their top left region of the grid.
We find these cells, color them white and iteratively repeat this process to find
the candidates for being the second, third, etc.\ black cell of the path.
We can prove that we need at least
$\left\lceil \frac{3 - \sqrt{5}}{4}n \right\rceil$ iterations until the grid is white.
%\todo[inline]{Emilio: STACS Reviewer 2: ``Mention that a similar lower bound technique (which counts the length of a monotone path in the matrix) has been given in the content of linear layouts in the following paper: Alam, Bekos, Gronemann, Kaufmann, Pupyrev, The mixed page number of graphs.''}
\end{proof}

Since $K_n$ contains $K_{\lfloor\frac{n}{2}\rfloor,\lfloor\frac{n}{2}\rfloor}$ as a subgraph, \cref{thm:k_nn} immediately implies a linear lower bound on~$\pqn(K_{n})$:

\begin{corollary}
	\label{cor:k_n}
	$\pqn(K_{n}) \geq \left\lfloor \frac{3-\sqrt{5}}{8}n \right\rfloor \approx 0.0955 n$.
\end{corollary}

\section{Characterizing Graphs with Priority Queue Number~1}\label{sec:pqn1}

We study which graphs have priority queue number~1.
We will provide a characterization of the graphs with priority queue number~1 in terms of forbidden minors (\cref{sec:pqn1:summary}).
We first describe families of graphs with priority queue number~1 (\cref{sec:pqn1:positive}).
Then, we give a set of eight forbidden minors (\cref{sec:pqn1:negative}).
The characterization of \cref{sec:pqn1:summary} follows by proving that every graph that does not contain any of the forbidden minors of \cref{sec:pqn1:negative} falls in one of the families of \cref{sec:pqn1:positive}.
This characterization leads to a linear-time recognition algorithm.

\subsection{Graphs with Priority Queue Number~1}\label{sec:pqn1:positive}

We begin by showing how to construct PQ-layouts of trees:

\begin{lemma}\label{lem:tree}
Let $\langle T=(V,E), w \rangle$ be an $n$-vertex edge-weighted tree with root $r \in V$. Then, a PQ-layout $\Gamma$ of $\langle T, w \rangle$ with $\pqn(\Gamma)=1$, where for each vertex $v \neq r$ its parent $p(v)$ precedes~$v$ in the linear ordering $\prec$ along the spine, can be constructed in $\mathcal{O}(n \log n)$ time.
\end{lemma}

\begin{proof}
We describe an algorithm that constructs a PQ-layout $\Gamma$ of $T$; see \cref{fi:pq-layout-tree} for an example.
For each vertex $v \in V \setminus \{r\}$, we define its weight $w(v)$ as $w(v)=w(vp(v))$.
As an initial layout, we first place $r$ as the leftmost vertex and the children of $r$ in increasing order of their weights to the right of $r$.
During our algorithm, we keep track of the \emph{last expanded vertex} $v^\star$, which is the rightmost vertex whose children (if any) have already been placed.
Note that initially $v^\star=r$.
We now maintain the following invariants:
\begin{enumerate}[{T.}1]
\item\label{inv:tree:1} For each vertex $v$ with $v \preceq v^\star$, all children have already been placed. 
\item\label{inv:tree:2} For each vertex $v$ with $v^\star \prec v$, no child has been placed yet.
In addition, $p(v) \preceq v^\star$.
\item\label{inv:tree:3} Let $S$ denote the set of already placed vertices succeeding $v^\star$ in $\prec$.
Then, the vertices in~$S$ occur in $\prec$ in increasing order of their weights.
\item\label{inv:tree:4} The already constructed PQ-layout has priority queue number~1.
\end{enumerate}
It is easy to see that the invariants hold for the initial layout.
We now iteratively consider the vertex $v'$ immediately succeeding $v^\star$ in $\prec$ and perform the following operations.
For each child $c$ of $v'$, we place $c$ to the right of $v'$ such that $S \cup \{c\} \setminus \{v'\}$ is ordered by weight.
This can be easily done since by \invref{T}{inv:tree:3}, $S$ (and thus also $S \setminus \{v'\}$) is already ordered by weight, i.e., there is a unique position for each $c$ and \invref{T}{inv:tree:3} holds again after $v'$ has become the new~$v^\star$.
Next, consider edge $v'c$.
By \invref{T}{inv:tree:4}, the only way that the resulting linear layout has not priority queue number~1 is if $v'c$ is included in a forbidden configuration.
Consider an edge $uv$ with $u \prec v$ that is co-active with $v'c$.
Clearly, $v \in S$ and by \invref{T}{inv:tree:2}, we have that $u \preceq v'$ and $u=p(v)$.
By \invref{T}{inv:tree:3}, it follows that $c$ and $v$ are sorted by their weights, which are equal to $w(v'c)$ and $w(uv)$, thus, $uv$ and $v'c$ form no forbidden configuration, i.e., \invref{T}{inv:tree:4} is maintained.

\begin{figure}
	\centering
	\begin{subfigure}[t]{0.45\textwidth}
		\centering
		\includegraphics[page=1]{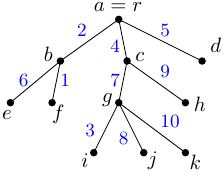}
		\subcaption{A weighted rooted tree $T$}
		\label{fi:pq-layout-tree-a}
	\end{subfigure}
	\hfil
	\begin{subfigure}[t]{0.45\textwidth}
		\centering
		\includegraphics[page=2]{pq-layout-tree}
		\subcaption{A PQ-layout $\Gamma$ of $T$.}
		\label{fi:pq-layout-tree-b}
	\end{subfigure}
	\caption{A PQ-layout of a weighted rooted tree computed with the algorithm in \cref{lem:tree}.%
        %The edge weights are in blue.
        }
	\label{fi:pq-layout-tree}
\end{figure}

After placing all children of $v'$, we set $v^\star = v'$.
Clearly, \invref{T}{inv:tree:1} is maintained as we placed all children of $v'$ whereas by induction the children of the vertices preceding $v'$ have already been placed.
Moreover, \invref{T}{inv:tree:2} is maintained as we only placed all children of $v'$ to the right of $v'$ whereas by induction the children of other vertices of $S$ are not placed yet.

For every vertex, we need to find its position depending on the weight.
This can be done in $\mathcal{O}(\log n)$ time if
we maintain a suitable search data structure.
The rest of the insertion can be done in constant time.
Overall, this results in a running time in $\mathcal{O}(n \log n)$.
%\todo[inline]{STACS Reviewer 2: ``Again, I am wondering whether the proof can be simplified by an inductive proof, which removes the heaviest (or lightest) edge, computes by induction a linear layout that satisfies the invariants and then shows how to introduce the removed edge so as to satisfy the invariants.''}
\end{proof}

Next, we construct a specific layout of caterpillars, which are a special subclass of trees.
Namely, a \emph{caterpillar} $C$ consists of an \emph{underlying path} $P(C)$ and additional leaves, each having exactly one neighbor on $P(C)$.
Be aware that, for the same caterpillar $C$, there are several choices for $P(C)$, and in particular $P(C)$ might end in a degree-$1$ vertex of $C$.
% \todo{Torsten: maybe add a note that there may be multiple
% choices of underlying path / note deviation from standard definition(?)}

\begin{lemma}\label{lem:caterpillar}
Let $\langle C, w \rangle$ be an $n$-vertex edge-weighted caterpillar with underlying path $P(C)$ and let $r$ be one of the two degree-$1$ vertices in $P(C)$.
Then, a PQ-layout $\Gamma$ of $\langle C, w \rangle$ with $\pqn(\Gamma)=1$ where $r$ is the rightmost vertex can be constructed in $\mathcal{O}(n)$ time.
\end{lemma}

\begin{proof}
We label $P(C)$ as $p_1p_2{\ldots}p_k=r$ and arrange the vertices of $P(C)$ in the order $p_1 \prec p_2 \prec \ldots \prec r$.
Then for $i \in \{1,\ldots,k\}$, we place all leaves at $p_i$ between $p_{i-1}$ and $p_i$ in an arbitrary order (note that for $p_1$, we place its leaves before $p_1$).
It is easy to see that the resulting layout contains %neither nestings nor pseudo-nestings nor crossings
no forbidden configuration and the statement follows.
\end{proof}

We now shift our attention to graphs containing cycles.
First, we consider cycles alone.

\begin{lemma}\label{lem:1pq-cycles}
Let $\langle C = (V,E), w \rangle$ be an $n$-vertex edge-weighted cycle, and let $v$ be any given vertex of~$V$.
Then, a PQ-layout $\Gamma$ of $\langle C, w \rangle$ with $\pqn(\Gamma)=1$ where $v$ is the leftmost vertex can be constructed in $\mathcal{O}(n)$ time.
\end{lemma}
\begin{proof}
%We first place $v$ as the leftmost vertex.
We maintain two \emph{candidate} vertices~$a$ and~$b$ to be placed next.
We call the neighbor of~$a$ ($b$, resp.) that has already been placed \emph{anchor} vertex~$a'$ ($b'$, resp.).
Initially, $a$ and $b$ are the two neighbors of the leftmost vertex~$v$ and $v = a' = b'$.
We iteratively append the candidate vertex whose edge to its anchor vertex has the smallest weight to the right of the linear order
(if both weights are equal, we take any of both candidates).
If $w(a'a) \le w(b'b)$, we append $a$ while we do not yet place~$b$.
After appending~$a$, we set $a' = a$ and the other neighbor of $a$ becomes the new candidate vertex~$a$.
We proceed symmetrically with~$b$ if $w(a'a) > w(b'b)$.
When $a = b$, we place the last remaining vertex $a$ as the rightmost vertex.

Clearly, every vertex is placed.
We show that no forbidden configuration occurs.
Suppose %for a contradiction
that there are two co-active edges $uv$ (where $u \prec v$) and $xy$ (where $x \prec y$) with $w(uv) > w(xy)$ and $v \prec y$.
In the course of our algorithm, after $x$ and $u$ have been placed and before $v$ and $y$ have been placed, the candidate vertices were $v$ and~$y$ whose anchor vertices were~$u$ and~$x$.
There could not have been a different anchor or candidate vertex because $C$ is a cycle.
Then, however, we would not have placed $v$ first because $w(uv) > w(xy)$~-- a contradiction.

Note that the running time of the initialization, termination and each iterative step is constant.
Therefore, the overall running time is in~$\mathcal{O}(n)$.
\end{proof}

We then consider the family of legged cycles.
Namely, a \emph{legged cycle} $L$ consists of a single \emph{underlying cycle} $C(L)$ and additional leaves, each having exactly one neighbor on $C(L)$.

\begin{restatable}[\restateref{lem:cycle}]{lemma}{Cycle}
\label{lem:cycle}
Let $\langle L, w \rangle$ be an $n$-vertex edge-weighted legged cycle with underlying cycle $C(L)$.
Then, a PQ-layout $\Gamma$ of $\langle L, w \rangle$ with $\pqn(\Gamma)=1$ can be constructed in $\mathcal{O}(n \log n)$~time.
\end{restatable}

\begin{proof}[Proof Sketch]
	We remove the heaviest edge~$e^\star$ of $C(L)$
	and all leaves where the weight of the incident edge is greater than $w(e^\star)$.
	We lay out the remaining graph using \cref{lem:caterpillar}
	such that the endpoints of $e^\star$ are the first and the last vertex.
	We reinsert~$e^\star$ and, in increasing order of the weights of their incident edges, we add the initially removed leaves to the right.
\end{proof}

%As we will see in \cref{sec:pqn1:negative}, 
%Attaching structures that are more complicated than just leaves to cycles will almost always result in priority queue number at least two.
%In fact, 
Next, we show that we may attach one caterpillar to any cycle (\cref{lem:one-caterpillar}) or two caterpillars to a single triangle or quadrangle (\cref{lem:two-caterpillars}).

\begin{lemma}\label{lem:one-caterpillar}
Let $\langle G,w \rangle$ be an $n$-vertex edge-weighted graph consisting of a cycle $O$ and a caterpillar~$C$ with underlying path $P(C)$ such that $V(O) \cap V(C) = \{r\}$ is a degree-$1$ vertex of $P(C)$.
Then, a PQ-layout $\Gamma$ of $\langle G, w \rangle$ with $\pqn(\Gamma)=1$ can be constructed in $\mathcal{O}(n)$ time.
\end{lemma}

\begin{proof}
We first draw $C$ and $O$ separately using \cref{lem:caterpillar,lem:1pq-cycles}, respectively.
Then, we identify the two occurrences of $r$ in $C$ and $O$ with each other, which yields our PQ-layout~$\Gamma$ of~$G$.
This does not result in a forbidden configuration since $r$ is the rightmost vertex of~$C$ and the leftmost vertex of~$O$ and, hence, there is no edge spanning over $r$.
\end{proof}

\begin{restatable}[\restateref{lem:two-caterpillars}]{lemma}{TwoCaterpillars}
\label{lem:two-caterpillars}
Let $\langle G, w \rangle$ be an $n$-vertex edge-weighted graph consisting of:
\begin{itemize}
\item A 3-cycle $\triangle = abc$ or a 4-cycle $\Box=abcd$.
\item A caterpillar $C_a$ with underyling path $P(C_a)$ starting at vertex $a$.
\item A caterpillar $C_c$ with underyling path $P(C_c)$ starting at vertex $c$.
\end{itemize}
Then, a PQ-layout $\Gamma$ of $\langle G, w \rangle$ with $\pqn(\Gamma)=1$ can be constructed in $\mathcal{O}(n \log n)$.
\end{restatable}

\begin{proof}[Proof Sketch]
	We arrange the cycle such that only the heaviest edge is co-active with the other edges of the cycle.
	We use \cref{lem:tree} for laying out the one caterpillar
	with designated leftmost vertex and we use \cref{lem:caterpillar}
	for laying out the other caterpillar with designated rightmost vertex.
	In the appendix, we argue that we can combine these three subgraphs,
	potentially adding some leaves later on, while still using only one priority queue.
\end{proof}

Finally, we consider graphs that contain more than one cycle.
If we exclude disconnected graphs, there exist exactly two such graphs with priority queue number~1, the complete bipartite graph $K_{2,3}$ and the graph obtained from  $K_4$ by removing an edge.
We will later see that these are the only ones.
The proofs that one priority queue suffices can be found in the appendix.

% \begin{lemma}\label{lem:k4minusE}
% $\pqn(K_4-e)=1$ where $K_4-e$ is obtained from $K_4$ by removing one edge.
% \end{lemma}
% 
% \begin{proof}
% \todo[inline]{add proof}
% \end{proof}

\begin{restatable}[\restateref{lem:k23}]{lemma}{KTwoThree}
\label{lem:k23}
$\pqn(K_{2,3})=1$.
\end{restatable}

\begin{restatable}[\restateref{lem:k4minusE}]{lemma}{KFourMinusE}
\label{lem:k4minusE}    
$\pqn(K_4-e) = 1$ where $K_4-e$ is obtained from $K_4$ by removing one edge.
\end{restatable}

\subsection{Forbidden Minors for Graphs with Priority Queue Number~1}\label{sec:pqn1:negative}

A \emph{minor}~$G'$ of a graph $G = (V,E)$ is obtained from $G$ by a series of edge contractions and by removing vertices and edges.
Formally, for an edge $e = uv$ in $G$, the \emph{contraction of edge $e$} yields the graph $G/e$, obtained by replacing vertices $u$ and $v$ by a single vertex $x_{uv}$ with incident edges $\{x_{uv}y \mid uy \in E \text{ or } vy \in E\}$.
Many important graph classes (e.g., planar graphs) are \emph{minor-closed}, that is,
if we take a minor of a graph lying in such a class, the resulting graph belongs to this class as well.
%Moreover, e
Every minor-closed graph class is defined by a finite set of forbidden minors and can be recognized efficiently~\cite{DBLP:journals/jct/RobertsonS95b,DBLP:journals/jct/RobertsonS04}.

Next, we prove that each of the graphs $\mathcal{F}_1,\ldots,\mathcal{F}_8$ shown in \cref{fig:minors} has priority queue number at least two.
The numbers at the edges specify edge-weight functions
for which one priority queue does not suffice.
Moreover, we prove that every graph that has one of $\mathcal{F}_1,\ldots,\mathcal{F}_8$ as a minor,
has priority queue number at least two, too.
%We verified this using a brute-force algorithm which attempted to draw the graphs with all permutations of their vertices.
%\todo{jz: we don't need this sentence if our proofs are quite clear}
%\todo[inline]{STACS Reviewer 2: `` To support this statement, you could somehow add in Figure 6 corresponding linear layouts showing that indeed these graphs has priority queue number 2.''}
In the appendix, we prove the following helpful observation:

\begin{restatable}[\restateref{le:1pq-cycle-last-vertex}]{lemma}{LastEdgeHeavy}
    \label{le:1pq-cycle-last-vertex}
	In any PQ-layout~$\Gamma$ of an edge-weighted cycle $\langle C, w \rangle$ with \mbox{$\pqn(\Gamma) = 1$}, the last vertex on the spine is incident to an edge of maximum weight in $\langle C, w \rangle$.
\end{restatable}

We now use \cref{le:1pq-cycle-last-vertex} to prove the main result of this subsection:

\begin{restatable}[\restateref{le:forbidden-minors-are-not-1pq}]{lemma}{ForbiddenMinors}
	\label{le:forbidden-minors-are-not-1pq}
    For every graph $G$ that has some of $\mathcal{F}_1,\ldots,\mathcal{F}_8$ (see \cref{fig:minors}) as a minor,
	$\pqn(G) > 1$ holds.
\end{restatable}
\begin{proof}[Proof of Cases $\mathcal{F}_1$--$\mathcal{F}_3$]
	To show $\pqn(G) > 1$,
	it suffices to show $\pqn(G, w) > 1$ for some edge weighting~$w$.
    For each $\mathcal{F} \in \{\mathcal{F}_1,\ldots,\mathcal{F}_8\}$,
    we prove that $\mathcal{F}$ itself has $\pqn(\mathcal{F}) > 1$
    and we prove that for any reversal of an edge contractions
    the resulting graph~$G$ has $\pqn(G) > 1$.
    In a reversal of an edge contraction,
    we split a vertex~$v$ into two vertices $v_1$ and~$v_2$, add the edge $v_1v_2$,
    and assign the edges incident to~$v$ arbitrarily to $v_1$ or $v_2$ or both.
    We call the reversal of an edge contraction a \emph{vertex split}.
    Note that it suffices to consider vertex splits
    since taking a supergraph of $\mathcal{F}$ can never reduce the
    number of priority queues needed for the edges of $\mathcal{F}$ alone.
    Further note that is suffices to consider vertex splits
    where each edge incident to $v$ is assigned to exactly one of~$v_1$ and~$v_2$
    and each of $v_1$ and~$v_2$ gets at least one incident edge of~$v$.
    If in $G$ this was not the case,
    we could remove duplicates of incident edges of~$v$
    and we could remove the one of~$v_1$ and~$v_2$ that has degree~1;
    the resulting subgraph $G'$ of~$G$ would have $\mathcal{F}$ as a minor.
    
	We claim that the edge weightings in \cref{fig:minors}
	require more than one priority queue.
	We consider each of the eight graphs and the graphs for which they are minors individually
	and suppose for a contradiction that there is a PQ-layout~$\Gamma$ of~$G$ with $\pqn(\Gamma) = 1$.
	We consider $\mathcal{F}_1$--$\mathcal{F}_3$ here
	and $\mathcal{F}_4$--$\mathcal{F}_8$ in the appendix.
	
	\begin{figure}
		\centering
		\begin{subfigure}[t]{0.1125\textwidth}
			\centering 
			\includegraphics[page=10]{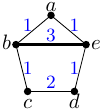}
			\subcaption{$\mathcal{F}_1$}
			\label{fig:minors:1}
		\end{subfigure}
		\hfill
		\begin{subfigure}[t]{0.1225\textwidth}
			\centering 
			\includegraphics[page=11]{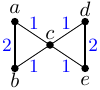}
			\subcaption{$\mathcal{F}_2$}
			\label{fig:minors:2}
		\end{subfigure}
		\hfill
		\begin{subfigure}[t]{0.125\textwidth}
			\centering 
			\includegraphics[page=12]{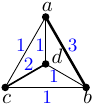}
			\subcaption{$\mathcal{F}_3$}
			\label{fig:minors:3}
		\end{subfigure}
		\hfill
		\begin{subfigure}[t]{0.1125\textwidth}
			\centering 
			\includegraphics[page=13]{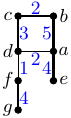}
			\subcaption{$\mathcal{F}_4$}
			\label{fig:minors:4}
		\end{subfigure}
		\hfill
		\begin{subfigure}[t]{0.1225\textwidth}
			\centering 
			\includegraphics[page=14]{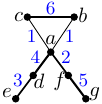}
			\subcaption{$\mathcal{F}_5$}
			\label{fig:minors:5}
		\end{subfigure}
		\hfill
		\begin{subfigure}[t]{0.145\textwidth}
			\centering 
			\includegraphics[page=15]{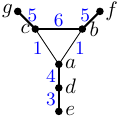}
			\subcaption{$\mathcal{F}_6$}
			\label{fig:minors:6}
		\end{subfigure}
		\hfill
		\begin{subfigure}[t]{0.1025\textwidth}
			\centering 
			\includegraphics[page=16]{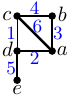}
			\subcaption{$\mathcal{F}_7$}
			\label{fig:minors:7}
		\end{subfigure}
		\hfill
		\begin{subfigure}[t]{0.1025\textwidth}
			\centering 
			\includegraphics[page=17]{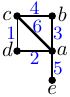}
			\subcaption{$\mathcal{F}_8$}
			\label{fig:minors:8}
		\end{subfigure}
		\caption{Forbidden minors for graphs with  priority queue number~1.
        The graphs with the  provided edge weights do not admit a linear layout with priority queue number~1.
        The thickness of an edge indicates the weight of that edge.}
		\label{fig:minors}
	\end{figure}
	
	\begin{description}
    \item[$\mathcal{F}_1$]:
        First consider $G = \mathcal{F}_1$.
        By \cref{le:1pq-cycle-last-vertex},
        $c$ or $d$ is the last vertex of the 5-cycle $abcde$.
        In the 4-cycle $bcde$, however,
        $b$ or $e$ is the last vertex; a contradiction.
        
        Now let $G$ be any graph that has $\mathcal{F}_1$ as a minor.
        As explained above, we do not need to consider vertex splits where new leaves or cycles arise.
        So, with a sequence of vertex splits, we can increase
        (i)~the length of the path~$bcde$,
        (ii)~the length of the path~$bae$, or
        (iii)~the length of the path~$be$.
        In any case, we assign any new edge a weight of~1.
        Then, our arguments about the last vertices also apply for larger cycles.
        
    \item[$\mathcal{F}_2$]:
        First consider $G = \mathcal{F}_2$.
        By \cref{le:1pq-cycle-last-vertex} and symmetry, we may assume that $a$ and $d$ are the last vertices of $\triangle abc$ and $\triangle cde$, respectively.
        Moreover, assume by symmetry that $a \prec d$, i.e., $c \prec a \prec d$.
        However, then edge $ab$ ends below the lighter edge $cd$; a contradiction.
        
        Now let $G$ be any graph that has $\mathcal{F}_2$ as a minor.
        With a sequence of vertex splits, we can
        (i)~increase the length of the cycles~$abc$ or~$cde$,
        (ii)~create a path (replacing~$c$) between the two cycles;
        we denote the two endpoints of that path as~$c_1$ and~$c_2$
        where $c_1$ is part of the one cycle and~$c_2$ is part of the other cycle.
        In any case, we assign any new edge a weight of~1.
        In~(i), we can still assume that $c \prec a \prec d$.
        In contrast to before $c$ might not be a neighbor of~$a$ or~$d$ any more.
        However, there is a path of edges with weight~1 between $c$ and~$a$
        and between~$c$ and~$d$.
        Some edge~$e$ on the path between $c$ and~$d$ spans over~$a$,
        which is the endpoint of the heavier edge~$ab$; a contradiction.
        In (ii), we establish again
        (by \cref{le:1pq-cycle-last-vertex} and symmetry up to the length of the cycles, which is not relevant here)
        that $c_1 \prec a$, $c_2 \prec d$, $a \prec d$.
        To obtain the contradiction that an edge of weight~1 on the path between~$c_2$ and~$d$
        spans over the endpoint of the edge~$ab$ of weight~2, we show that $c_2 \prec a$.
        Suppose for a contradiction that $c_1 \prec a \prec c_2$.
        Then, there is some edge of weight~1 on the path between~$c_1$ and~$c_2$
        that spans over the heavier edge $ab$; a contradiction.
        
	\item[$\mathcal{F}_3$]:
        First consider $G = \mathcal{F}_3$.
		Consider the triangle $\triangle abd$.
		By \cref{le:1pq-cycle-last-vertex},
		$a$ or $b$ is the last vertex of $\triangle abd$ (on the spine)
		since $ab$ is its heaviest edge.
		In the triangle $\triangle bcd$, $cd$ is the heaviest edge,
		which implies that $c$ or $d$ is the last vertex of $\triangle bcd$.
		Hence, neither $b$ nor $d$ can be the (overall) rightmost vertex.
		Since the heaviest edge of the 4-cycle $abcd$
		is again $ab$, $a$ is the rightmost vertex.
		In the triangle $\triangle acd$, $cd$ is the heaviest edge,
		which implies that $a$ cannot be the rightmost vertex;
		a contradiction.
        
        Now let $G$ be any graph that has $\mathcal{F}_3$ as a minor.
        Since $\mathcal{F}_3$ is the complete graph on four vertices (i.e., $K_4$),
        it is highly symmetric and the first vertex split makes one of the
        four 3-cycles a 4-cycle.
        This new graph is precisely $\mathcal{F}_1$ with one additional edge.
		\qedhere
	\end{description}
\end{proof}

\subsection{Characterization and Recognition}\label{sec:pqn1:summary}

We will now combine the results shown in the previous subsections to show the following:

\begin{restatable}[\restateref{thm:pq1}]{theorem}{pqOne}
    \label{thm:pq1}
    Let $G$ be a graph. Then, $\pqn(G)=1$ if and only if $G$ does not contain any~$\mathcal{F}_i$ for $i \in \{1,\ldots,8\}$ as a minor.
\end{restatable}

\begin{proof}[Proof Sketch]
    We analyze the structure of~$G$ and we see that either
    $G$ is a graph that has $\pqn(G) = 1$ due to a result from \cref{sec:pqn1:positive},
    or $G$ has some $\mathcal{F}_i$ for $i \in \{1,\ldots,8\}$ as a minor.
    
    If $G$ is a tree, then $\pqn(G) = 1$ by \cref{lem:tree}.
    If $G$ contains exactly one cycle~$C$,
    then $G - E(C)$ is a forest $F$, and we consider each component of $F$ as rooted at~$C$.
    If every component is an isolated vertex or a rooted star, then $\pqn(G) = 1$ by \cref{lem:cycle}.
    If some component is more complex than a rooted star or a rooted caterpillar, then $G$ contains $\mathcal{F}_5$ as a minor.
    Now suppose some component~$K$ is a rooted caterpillar.
    If all other components are isolated vertices, then $\pqn(G)=1$ by \cref{lem:one-caterpillar}.
    If there are two further non-trivial components, then $G$ contains $\mathcal{F}_6$ as a minor.
    So assume that there is exactly one further non-trivial component~$K' \neq K$.
    If the roots of $K$ and $K'$ have distance at least $3$ in one direction along $C$, then $G$ contains $\mathcal{F}_4$ as a minor.
    If none of the above applies, either $C$ is a triangle and $\pqn(G)=1$ by \cref{lem:two-caterpillars},
    or $C$ is a quadrangle and the roots of $K$ and~$K'$ are opposite on $C$, and $\pqn(G) = 1$ by \cref{lem:two-caterpillars}.
    
    The case where~$G$ contains at least two cycles remains.
    If $G$ has two edge-disjoint cycles, then $G$ contains $\mathcal{F}_2$ as a minor.
    If two cycles have at least two vertices but no edge in common, then $G$ contains $\mathcal{F}_8$ as a minor.
    If any two cycles share at least one edge,
    then either $G = K_{2,3}$ or $G = K_4-e$ and we have $\pqn(G) = 1$ by \cref{lem:k23,lem:k4minusE},
    or $G$ contains $\mathcal{F}_1$, $\mathcal{F}_3$, $\mathcal{F}_7$, or $\mathcal{F}_8$ as a minor.
\end{proof}

\begin{corollary}\label{cor:minorClosed}
    The family of graphs with priority queue number 1 is minor-closed.
\end{corollary}

\cref{thm:pq1} implies that the structure of a graph $G$ with $\pqn(G)=1$ is quite limited. 
In fact, forbidden minors $\mathcal{F}_1, \mathcal{F}_2$ and $\mathcal{F}_3$ imply that $G$ can contain at most three cycles. 
More precisely, if $G$ contains two cycles, it is necessarily a $K_{2,3}$ or a $K_4-e$ since forbidden minors $\mathcal{F}_7$ and $\mathcal{F}_8$ forbid any other edge. 
Moreover, if $G$ contains a single cycle, forbidden minor $\mathcal{F}_4$ implies that no non-trivial tree (i.e., a tree that is more than a caterpillar) can be attached to the cycle. Finally, forbidden minors $\mathcal{F}_4$ and $\mathcal{F}_6$ dictate how legs can be combined with an attached caterpillar.
$\mathcal{F}_6$ implies that at most one vertex of the cycle can have legs whereas $\mathcal{F}_4$ implies that, for a cycle of length $4$, legs can only be attached to the vertex opposite of the vertex attached to the caterpillar whereas for cycles of lengths larger than $4$, no legs are allowed.
These observations in conjunction with \cref{lem:tree,lem:caterpillar,lem:cycle,lem:two-caterpillars,lem:k23,lem:k4minusE} imply the following:

\begin{corollary}\label{cor:recognition}
Given a graph $G$, it can be decided in $\mathcal{O}(n)$ time if $\pqn(G)=1$. Moreover, if $\pqn(G)=1$, a PQ-layout $\Gamma$ of $\langle G, w \rangle$ with $\pqn(\Gamma)=1$ can be computed in $\mathcal{O}(n \log n)$ time for a given edge-weight function $w$.
\end{corollary}

%%%%%%%%%%
\section{PQ-Layouts of Graphs with Bounded Pathwidth and Treewidth}\label{sec:results-kpq}

We present results for graphs of bounded pathwidth and bounded treewidth.
We assume familiarity with these concepts;
otherwise see \cref{app:pathwidth-treewidth} for formal definitions.

\begin{theorem}
    Let $G$ be a graph with pathwidth at most $p$. Then, $\pqn(G) \leq p+1$.
%    \todo{Reviewer~1: ``Is this bound tight?''}
\end{theorem}

\begin{proof}
    We may assume w.l.o.g.\ %without loss of generality
    that $G = (V,E)$ is edge-maximal of pathwidth $p$, i.e., $G$ is an interval graph with clique number $\omega(G) = p+1$~\cite{DBLP:journals/tcs/Bodlaender98}.
    Let $\{I_v = [a_v,b_v]\}_{v \in V}$ be an interval representation of $G$ with distinct interval endpoints.
    In particular $b_u \neq b_v$ for any $u \neq v \in V$.
    
    We take the ordering $v_1,\ldots,v_n$ of vertices given by their increasing right interval endpoints.
    That is, $b_{v_1} < \cdots < b_{v_n}$.
    Crucially, for any $i < j < k$ with $v_iv_k \in E$, also $v_jv_k \in E$.% is an edge in $G$.
    
    To define the partition of $E$ into $p+1$ priority queues $\mathcal{P}_1,\ldots,\mathcal{P}_{p+1}$, we consider a proper vertex coloring of $G$ with $\chi(G) = \omega(G) = p+1$ colors.
    This exists, as $G$ is an interval graph.
    Now for $c = 1,\ldots,p+1$ let $\mathcal{P}_c$ be the set of all edges in $G$ whose right endpoint in the vertex ordering has color $c$.
    In fact, each $\mathcal{P}_c$ contains none of the forbidden configurations in \cref{fi:forbidden-configurations}, and hence is a priority queue independent of the edge-weights, since in each forbidden configuration in \cref{fi:forbidden-configurations} there would be an edge in $G$ between the right endpoints $v_r$ and $v_{r'}$ of the two edges $e$ and $e'$ forming the forbidden configuration.
	But this would imply different colors for $v_r$ and $v_{r'}$ and therefore different priority queues for $e$ and $e'$.
\end{proof}

We remark that \cref{cor:k_n} implies that there are graphs with pathwidth $p$ and priority queue number at least $\left\lfloor \frac{3-\sqrt{5}}{8}(p+1)\right\rfloor$ as $K_{p+1}$ has pathwidth $p$.

Next, we shall construct for every integer $p \geq 1$ a graph $G = G_p$ of treewidth~$2$ together with edge weights $w$ so that every vertex ordering $\sigma$ of $G$ contains a $p$-inversion.
If
%Recall that if for a particular edge-weighting $w$ of $G$,
every vertex ordering of $\langle G,w \rangle$ contains a $p$-inversion, then $\pqn(G,w) \geq p$ and hence $\pqn(G) \geq p$.

\begin{restatable}[\restateref{thm:treewidth-2}]{theorem}{TreewidthTwo}
	\label{thm:treewidth-2}
    For every $p$ there is a graph $G_p$ of treewidth~$2$ and $\pqn(G_p) \geq p$.
\end{restatable}

\begin{proof}[Proof Sketch]
    Let $p \geq 1$ be a fixed integer.
    Our desired graph $G_p$ will be a \emph{rooted $2$-tree}, i.e., an edge-maximal graph of treewidth~$2$ with designated root edge, which are defined inductively through the following construction sequence:
    \begin{itemize}
        \item A single edge $e^* = v'v''$ is a $2$-tree rooted at $e^*$.
        \item If $G$ is a $2$-tree rooted at $e^*$ and $e = uv$ is any edge of $G$, then adding a new vertex $t$ with neighbors $u$ and $v$ is again a $2$-tree rooted at $e^*$.
        We say that $t$ is \emph{stacked onto~$uv$}.
%             In this case we say that $t$ is \emph{stacked onto the edge $uv$}, and vertices $u$ and $v$ are called the \emph{parents} of $t$.
    \end{itemize} 
    
    We define a vertex ordering $\sigma$ of a $2$-tree rooted at $e^* = v'v''$ to be \emph{left-growing} if no vertex (different from $v'$, $v''$) appears to the right of both its parents.
    First we show that it is enough to force a $p$-inversion (and hence $\pqn(G) \geq p$) in every left-growing vertex ordering.
    
    \begin{restatable}[\restateref{claim:wlog-left-growing}]{claim}{LeftGrowing}
    	\label{claim:wlog-left-growing}
        If $\langle G,w \rangle$ is an edge-weighted rooted $2$-tree so that every left-growing vertex ordering contains a $p$-inversion, then there exists an edge-weighted rooted $2$-tree $\langle \tilde{G},\tilde{w} \rangle$ so that every vertex ordering contains a $p$-inversion.
    \end{restatable}
    
    \begin{claimproof}[Proof Sketch]
        Intuitively, whenever in the construction sequence of $G$ a vertex $t$ is stacked onto an edge $uv$, in $\tilde{G}$ we instead stack $p^2$ ``copies'' of $t$ called $t_1,\ldots,t_{p^2}$ onto $uv$ (treating each $t_i$ like $t$ henceforth).
        The weights $\tilde{w}(ut_i)$ and $\tilde{w}(vt_i)$, $i=1,\ldots,p^2$, are chosen very close to $w(ut)$ and $w(vt)$ but increasing with $i$ for $ut_i$ and decreasing with $i$ for $vt_i$.
        This way, if all $p^2$ copies of $t$ lie right of $u$ and $v$, we get a $p$-inversion among the $ut_i$'s or $vt_i$'s.
    \end{claimproof}

	\begin{figure}[t]
		\centering
		\includegraphics{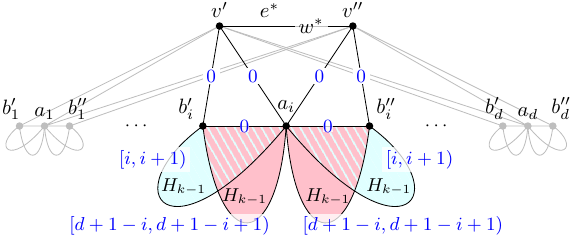}
		\caption{Illustration of the $2$-tree $H_k$ with edge-weighting $w_k$, for $k\geq2$.}
		\label{fig:2-trees}
	\end{figure}

    We then construct a sequence $\langle H_1,w_1 \rangle, \langle H_2,w_2 \rangle, \ldots$ of edge-weighted rooted $2$-trees, so that for every $k$, every left-growing vertex ordering of $\langle H_k,w_k \rangle$ contains a $k$-inversion.
    The construction of $\langle H_k,w_k \rangle$ is recursive, starting with $\langle H_1,w_1 \rangle$ being just a single edge of any weight.
    For $k \geq 2$, we start with $d=k^2$ vertices $a_1,\ldots,a_d$ stacked onto the root edge $e^* = v'v''$, and stacking a vertex $b'_i$ onto each $v'a_i$, as well as a vertex $b''_i$ onto $v''a_i$ for $i=1,\ldots,d$.
    All these edges have weight~$0$.
    Then, each $b'_ia_i$ and $b''_ia_i$ is used as the base edge of two separate copies of $\langle H_{k-1},w_{k-1} \rangle$, where however the edge-weights $w_{k-1}$ of each copy are scaled and shifted to be in a specific half-open interval $[x,y) \subset \mathbb{R}$ depending on the current base edge ($b'_ia_i$ or $b''_ia_i$).
    See \cref{fig:2-trees} for an illustration.

    By induction, there is a $(k-1)$-inversion in each such copy of $H_{k-1}$.
    Using the pigeon-hole principle, and the Erd\H{o}s--Szekeres theorem, we then identify an edge $a_iv'$ (or $a_iv''$) which forms a $k$-inversion either together with $(k-1)$-inversion of one copy of $H_{k-1}$, or with $k-1$ edges, each from the $(k-1)$-inversion of a different copy of $H_{k-1}$.
    
    To finish the proof, it is then enough to take $\langle H_p,w_p \rangle$ as constructed above, and then apply \cref{claim:wlog-left-growing} to obtain the desired edge-weighted $2$-tree $\langle G_p,w_p \rangle$ with $\pqn(G_p,w_p) \geq p$.
\end{proof}

As graphs of treewidth $2$ are always planar graphs, we conclude with the following result that contrasts similar results for the queue and stack number~\cite{DBLP:journals/jacm/DujmovicJMMUW20,DBLP:journals/jcss/Yannakakis89}:

\begin{corollary}
The priority queue number of planar graphs is unbounded. In particular, there is a planar graph $G$ with $n$ vertices and $\pqn(G) \geq \log_4 n$.
\end{corollary}

\begin{proof}
We observe that $H_k$ in the proof of \cref{thm:treewidth-2} contains $4$ copies of $H_{k-1}$ while $H_1$ contains $2$ vertices. Thus, there are at least $4^p$ vertices in $G := H_p$, which has treewidth~$p$.
\end{proof}

%%%%%%%%%%
\section{Complexity of Fixed-Ordering PQ-Layouts}
\label{sec:complexity}

In the light of \cref{cor:recognition}, one might wonder whether
deciding if a given graph~$G$ has priority queue number~$k$ is polynomial-time solvable for all values of~$k$.
However, we show that it is \NP-complete if the ordering of the vertices is already fixed.

% in this section, we show that deciding whether $k$ priority queues suffice for an edge-weighted graph is \NP-hard if the ordering of the vertices is already fixed by the input.
%\todo{jz: now $V(G)$ is eliminated, but we still use $V(H)$; check and maybe change back and introduce $V(G)$;
%also used in \cref{lem:one-caterpillar}}

\begin{theorem}
	\label{thm:npc-fixed-order}
	Given an edge-weighted graph $\langle G, w \rangle$ with $G=(V,E)$,
	a linear ordering $\sigma$ of $V$, 
	and a positive integer~$k$, it is \NP-complete to decide
	whether $\langle G, w \rangle$ admits a PQ-layout with linear ordering $\sigma$ and a partitioning of $E$ into $k$ priority queues.
\end{theorem}

%\todo[inline]{Johannes: STACS Reviewer 2: `` Somewhere in the statement of the theorem mention that k is not fixed (as mentioned in l.500). Regarding exactly this point, isn’t it that the circular-arc coloring problem is NP-complete for fixed values of k>3 (Unger, (1988), On the k-colouring of circle-graphs). In this case, why doesn’t your proof hold for k>3 as well? Am I missing something?''
%
%HF: I made this into an open problem for now. We may think about writing it awkwardly into the theorem.
%Also in the proof we may point out that circular-arc graphs and circle graphs are something different!}
%
%\todo[inline]{Johannes: this would require more thinking; postponed.
%STACS Reviewer 1: ``it is well known that 4-coloring of circular graphs is NP-complete (Unger, STACS’88). Can this result strengthen the claim of Thm 24?''}

\begin{proof}
	Containment in \NP{} is clear as we can check in polynomial time
	whether a given assignment of edges to $k$ priority queues
	corresponds to a valid PQ-layout.
	
	We show \NP-hardness by reduction from the problem
	of coloring circular-arc graphs, which is known
	to be \NP-complete if $k$ is not fixed~\cite{DBLP:journals/siammax/GareyJMP80}.
	A circular-arc graph is the intersection graph
	of arcs of a circle (see \cref{fig:nph-fixed-order} on the left),
	that is, the arcs are the vertices and two vertices share
	an edge if and only if the corresponding arcs share
	a point along the circle.
	
	For a given circular-arc graph~$H$, we next describe
	how to obtain an edge-weighted graph $\langle G, w \rangle$
	and a vertex ordering~$\sigma$
	such that a $k$-coloring of $H$ directly corresponds
	to a PQ-layout of $\langle G, w \rangle$
	under~$\sigma$ with $k$ priority queues.
	Without loss of generality, we assume
	that we have a circular-arc representation of~$H$
	where all endpoints of the arcs are distinct.
	
	We ``cut'' the circle at some point,
	which gives an interval representation
	where some vertices $S$ of~$H$ refer to two intervals;
	in \cref{fig:nph-fixed-order}, the circular arcs
	$a$, $b$, and $d$ are cut into intervals $a_1$ and $a_2$, $b_1$ and $b_2$, and $d_1$ and $d_2$, respectively.
	We let the endpoints of all intervals be the vertices of $G$,
	whose ordering $\sigma$ along the spine is taken from the
	ordering of the corresponding endpoints of the intervals
	(for the intervals belonging to $S$, the order is arbitrary).
	Further, for each interval, we have an edge in~$G$
	connecting its two endpoints.
	To obtain~$w$, we assign to the edges of~$G$ the numbers 1 to $m$,
	where $m = |V(H)| + |S|$, in decreasing order
	of the right endpoints of the intervals.
	This way, every pair of edges whose corresponding intervals share a common
	point need to be assigned to distinct priority queues.
	
	\begin{figure}[t]
		\centering
		\includegraphics[scale=.95]{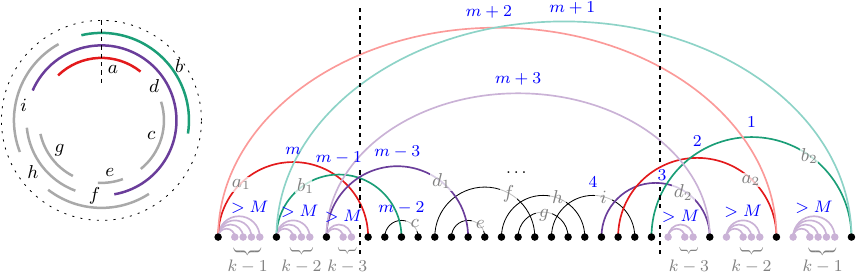}
		\caption{Illustration of our reduction.}% from coloring
		%	circular-arc graphs (example instance on the left)
		%	to the problem of determining,
		%	for an edge-weighted graph with fixed vertex ordering,
		%	the number of priority queues (resulting instance on the right).}
		\label{fig:nph-fixed-order}
	\end{figure}
	
	So far, there is no mechanism that assures that $a_1$ and~$a_2$,
	$b_1$ and $b_2$, etc., are assigned to the same priority queue.
	To this end, we add, for each such pair $\langle a_1, a_2 \rangle$,
	an edge from the left endpoint of~$a_1$ to the right endpoint of~$a_2$.
	We call these $|S|$ new edges \emph{synchronization edges},
	and we assign them the weights $m+1, \dots, m+|S|$
	in decreasing order of their right endpoints.
	Furthermore, we add $k-1, k-2, \dots$ \emph{heavy} edges
	below the leftmost, second leftmost, \dots~vertex on the spine, respectively,
	as well as, $k-1, k-2, \dots$ heavy edges
	below the rightmost, second rightmost, \dots~vertex on the spine, respectively,
	see \cref{fig:nph-fixed-order}.
	Clearly, $k \ge |S|$, otherwise we have a no-instance.
	The heavy edges have distinct right endpoints and their weights are chosen
	such that (i)~they are greater than $M$,
	where $M = m + |S|$ is the largest weight used so far, and
	(ii)~they are chosen in decreasing order
	from left to right such that each pair of heavy edges needs
	to be assigned to distinct priority queues if they overlap.
	
	\begin{restatable}[\restateref{claim:pair-same-queues}]{claim}{PairSameQueues}
		\label{claim:pair-same-queues}
		Let $\Gamma$ be a PQ-layout of $\langle G, w \rangle$
		under vertex ordering~$\sigma$ having $k$ priority queues.
		For each vertex $a \in S$ ($\subseteq V(H)$),
		the two corresponding edges $a_1$ and $a_2$ in~$G$
		are assigned to the same priority queue in~$\Gamma$.
	\end{restatable}
	\begin{claimproof}[Proof Sketch]
		Consider the vertices of $G$ from the outside to the inside.
		At the outer endpoint of $a_1$ ($a_2$, resp.), the edge $a_1$ ($a_2$, resp.)
		and the synchronization edge get the same color
		by the pigeonhole principle because
		all but one priority queue is occupied by the heavy edges
		and the previously considered edges,
		which all induce pairwise forbidden configurations.
	\end{claimproof}
	
	If we have a PQ-layout of $\langle G, w \rangle$
	under vertex ordering~$\sigma$ with $k$ priority queues,
	we obtain, due to the choice of $w$ and $\sigma$,
	a coloring of~$H$ by using the priority queues as colors.
	In particular, due to \cref{claim:pair-same-queues},
	both parts of the ``cut'' circular arcs get the same color.
	Conversely, we can easily assign the edges of $\langle G, w \rangle$
	to $k$ priority queues if we are given a $k$-coloring of~$H$.
	Clearly,
%	As
	our reduction can %clearly
	be implemented to run in polynomial time.%,
%	this concludes the correctness of the theorem.
\end{proof}

%%%%%%%%%%
\section{Open Problems}\label{sec:conclusions}

% In this paper, we introduced PQ-layouts as a natural linear layout variant for edge-weighted graphs. Inspired by the existing literature on stack and queue layouts, we investigated typical topics of interest; yet several intriguing research questions remain open.  

% First, we showed that complete graphs have linear priority queue number however, our lower bound is obtained by considering only a complete bipartite subgraph. Thus, it would be an intriguing problem to obtain a better lower bound for general complete graphs. Moreover, a tight upper bound is still elusive.

%In this paper, we introduced the priority queue number of graphs as a novel linear-layout variant.
We conclude by summarizing a few intriguing problems that remain open.
%\begin{itemize}
%\item
(i) Improve the upper and lower bounds for $\pqn(K_n)$.
%\item
(ii) Determine the %computational
complexity of deciding $\pqn(G) \leq k$ for $k > 1$, or deciding $\pqn(G,w) \leq k$ for a given edge weighting~$w$ when $k$ is a constant
or no vertex ordering is given.
% \item
(iii) Find graph families with bounded priority queue number; what about outerplanar graphs?
%\item
(iv) Can the edge density of a graph~$G$ be upper bounded by a term in $\pqn(G)$?
%\end{itemize}
%\todo[inline]{All: STACS Reviewer 1: ``Graphs with a bounded stack/queue/etc number are typically sparse, which is helpful for deriving corresponding lower bounds. Is something known about the density of graphs with a bounded PQ-number? Is it a good question to study?''
%HF: Corollary 23 shows that the answer is less than $3n-6$. \\
%HF \& JZ: the implication from Cor. 23 is in the wrong direction.}
(v) One can also investigate the structure of edge-weighted graphs $\langle G,w\rangle$ with a fixed vertex-ordering $\prec$.
Edge-partitions of $G$ into $k$ priority queues are equivalent to proper $k$-colorings of an associated graph $H$, while $t$-inversions correspond to $t$-cliques in $H$.
Is $\chi(H)$ bounded by a function in terms of $\omega(H)$?
Is determining $\omega(H)$ \NP-complete?

\bibliography{literature}

\clearpage
\appendix
\section*{Appendix}

%%%%%%%%%%
\section{Omitted Proofs of \cref{sec:complete}}

\Knn*
\label{thm:k_nn*}

\begin{proof}
	Let $A$ and $B$ denote the two independent sets of size $n$.
	We partition the edge set into $n$ perfect matchings $M_1,\ldots,M_n$.
	For each edge $e \in M_i$ we set $w(e)=i$.
	Thus, each vertex is incident to exactly one edge of weight $i$ for $i \in \{1,\ldots,n\}$.
	
	Let $\Gamma$ be any PQ-layout of $K_{n,n}$.
	We first show that there exists a subgraph $K_{\frac{n}{2},\frac{n}{2}}$ with independent sets $A' \subseteq A$ and $B' \subseteq B$ of size $\frac{n}{2}$ such that the restriction $\Gamma'$ of $\Gamma$ to the subgraph induced by $A' \cup B'$ is a \emph{separated} PQ-layout.\footnote{%
		For simplicity, we assume that $n$ is even.}
	Let $v_1,\ldots,v_t$ be the first $t$ vertices in the linear ordering of $\Gamma$, such that $|\{v_1,\ldots,v_t\} \cap A| = \frac{n}{2}$ and $t$ is minimum.
	If $|\{v_{t+1},\ldots,v_{2n}\}\cap B| \geq \frac{n}{2}$, we set $A'=\{v_1,\ldots,v_t\} \cap A$ and $B'=\{v_{q},\ldots,v_{2n}\}\cap B$ where $q\geq t+1$ is chosen so that $|B'|=\frac{n}{2}$.
	Otherwise, by the choice of $t$, we have that $|\{v_{t+1},\ldots,v_{2n}\} \cap A| = \frac{n}{2}$, and, by the pigeon-hole principle, also $|\{v_1,\ldots,v_t\} \cap B| > \frac{n}{2}$.
	In this case, we choose $B'=\{v_1,\ldots,v_p\} \cap B$ where $p<t$ is chosen so that $|B'|=\frac{n}{2}$ and $A'=\{v_{t+1},\ldots,v_{2n}\} \cap A$.
	This proves the existence of~$\Gamma'$.
	
	%Let $\Gamma'$ be the restriction of $\Gamma$ to the $K_{\frac{n}{2}, \frac{n}{2}}$ induced by $A'$ and $B'$ defined above.
	%For what we have proved above, $\Gamma'$ is a bipartite PQ-layout.
	%of $K_{\frac{n}{2},\frac{n}{2}}$ defined above, with the independent sets $A'$ and $B'$ of size $\frac{n}{2}$ in which the vertices in $B'$ succeed the vertices in $A'$. 
	
	We now focus on $\Gamma'$.
	Without loss of generality, we can assume that $A'$ precedes $B'$ in the linear ordering of $\Gamma'$.
	Let $b_1,\ldots,b_{\frac{n}{2}}$ denote the linear ordering of the vertices of $B'$.
	Since $\Gamma'$ is a separated PQ-layout and since each of its edges has an endpoint in $A'$ and an endpoint in $B'$, %the intersection of the active intervals of any two edges of $\Gamma'$ contains the (non-empty) interval between $A'$ and $B'$. Therefore,
	all edges of $\Gamma'$ are co-active at~$b_1$.
	As $K_{\frac{n}{2},\frac{n}{2}}$ is a subgraph of $K_{n,n}$, each vertex in $B'$ is incident to exactly $\frac{n}{2}$ different edge weights occurring in the range $[1,n]$. Moreover, an edge $e$ incident to $b_i$ and an edge $e'$ incident to $b_j$ form a forbidden configuration if $i < j$ but $w(e) > w(e')$.
	
	\begin{figure}
		\centering
		\begin{subfigure}[t]{0.26\textwidth}
			\centering
			\includegraphics[page=1]{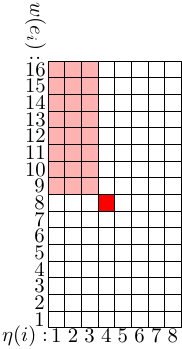}
			\subcaption{The dark red cell represents the edge of weight~8,
				having vertex $b_4$ as its right endpoint.
				The red shaded cells represent the region~$\mathcal{R}(e_i)$.}
			\label{fig:k_nn:1}
		\end{subfigure}
		\hfill
		\begin{subfigure}[t]{0.31\textwidth}
			\centering
			\includegraphics[page=2]{k_nn}
			\subcaption{Let the black and red cells be set to~1.
				The red cells constitute a strictly monotonically decreasing path.
				Together, they represent a $k$-inversion in the corresponding edge-weighted graph.}
			\label{fig:k_nn:2}
		\end{subfigure}
		\hfill
		\begin{subfigure}[t]{0.26\textwidth}
			\centering
			\includegraphics[page=3]{k_nn}
			\subcaption{Let the black, red, and blue cells be set to~1.
				The red cells represent the candidates for edge~$e_1$, while the blue edges represent the candidates for edge~$e_2$.}
			\label{fig:k_nn:3}
		\end{subfigure}
		\caption{Illustrations for the proof \cref{thm:k_nn}.
			In this matrix of size $n \times \frac n2$,
			all possible edge weights appear in the rows and the
			indices of the vertices $b_1, \dots, b_{\frac{n}{2}}$ appear in the columns.}
		\label{fig:k_nn}
	\end{figure}
	
	We now show that for some $k$ there is a $k$-inversion,
	that is, there is a sequence of edges $e_1,\ldots,e_k$ such that for $1 \leq i < j \leq k$, $w(i) > w(j)$ and $b_{\eta(i)}<b_{\eta(j)}$ where $\eta(i)$ is the index of the endpoint of edge $e_i$ in $B'$.
	Another way to see this is by visualizing the relation between $\eta(i)$ and $w(e_i)$ as a binary matrix\footnote{A similar, although not identical, representation is used by Alam et al.~\cite{MDALAM2022131} to prove the mixed page number of $K_{n,n}$.} $M$ of size $n \times \frac{n}{2}$ where $M[w,\eta]=1$ if there is an edge $e_i$ in~$\Gamma'$ with $w(e_i)=w$ and $\eta(i)=\eta$, and $M[w,\eta]=0$ otherwise; see \cref{fig:k_nn}.
	In this view of the problem, for a given edge $e_i$ (see the red cell in \cref{fig:k_nn:1}), forbidden configurations are created with \emph{any edge} whose cell (with value $1$) occurs in the region $\mathcal{R}(e_i)$ to the top-left of the cell of $e_i$ (see the red shaded cells in \cref{fig:k_nn:1}).
	Moreover, each column of the matrix has exactly $\frac{n}{2}$ cells with value $1$ (corresponding to the $\frac{n}{2}$ edges incident to each vertex in $B'$).
	Our desired sequence $e_1,\ldots,e_k$ thus is a \textit{strictly monotonically decreasing path} through the matrix along cells with value $1$; see the red path in \cref{fig:k_nn:2}.
	It is easy to see that each pair of edges along the sequence creates a forbidden configuration, i.e., $k$ priority queues are necessary.
	
	To compute a strictly monotonically decreasing path of maximum length, we proceed iteratively as follows.
	At step~1 we compute a subset of candidate edges (cells) for $e_1$, containing to every cell of an edge~$e$ for which $\mathcal{R}(e)$ is empty (i.e., it does not contain cells of value 1).
	At step~$i$, with $i \geq 2$, we compute a subset of candidate edges (cells) for $e_i$, containing every cells of an edge~$e$ for which $\mathcal{R}(e)$ is empty after all cells of the candidate edges for $e_{j}$ with $j < i$ have been removed (i.e., are set to~0 in the matrix).
	For example, in \cref{fig:k_nn:3}, the red cells correspond to the candidate edges for~$e_1$, while the blue cells correspond to the candidate edges for~$e_2$.  
	%At step $i$, we compute all the \emph{candidate edges} for edge $e_i$. Namely, all such edges are those for which $\mathcal{R}(e_i)$ is empty; see the red colored cells in \cref{fig:k_nn:3}.
	%Afterwards, we remove all candidates for $e_i$ from $M$; the candidates for $e_{i+1}$ are now all the edges for which the region $\mathcal{R}(e_{i+1})$ is empty; see the blue colored cells in \cref{fig:k_nn:3}. 
	We iteratively repeat the process until all $\frac{n^2}{4}$ cells with value $1$ have been set to~0.
	The number~$k$ of steps in this process yields the length of the sequence and hence the required number of priority queues.
	
	We now compute the value of $k$.
	To this end, we first prove that, in step~$i$, the topmost $i-1$ rows and the leftmost $i-1$ columns are empty (i.e., all cells are set to~0).
	The proof is by induction on $i$.
	When $i=1$ the claim is trivially true.
	Assume that the claim is true for $i-1$, with $i \geq 2$.
	For each edge~$e$ whose cell is at column~$i$, the region $\mathcal{R}(e)$ contains only cells at columns that are to the left of column~$i$, which are empty by the inductive hypothesis.
	Hence, all cells at column~$i$ are set to~0 in step~$i$.
	An analogous argument applies for the topmost $i-1$ rows.  
	%To this end, observe that by an inductive argument, in iteration $i$, the top-most $i-1$ rows are already empty.
	%Namely, this is trivially the case for step $1$ and in step $i$, we delete the cells for which $\mathcal{R}(e_i)$ is empty, which is true by induction hypothesis.
	%The same holds for the leftmost $i-1$ columns.
	Moreover, denote by $M'$ the sub-matrix of $M$ obtained from~$m$ after removing the $i-1$ leftmost columns and the $i-1$ topmost rows.
	In each diagonal of $M'$ of slope $-1$, at most one non-empty cell has an empty region to the top-left.
	Observe that $M'$ has $n-i+1$ rows and $\frac{n}{2}-i+1$ columns; thus the number of such diagonals is $\frac{3n}{2}-2i+1$ and the number of cells set to~0 at step~$i$ is at most $\frac{3n}{2}-2i+1$. 
	%Summing over all $k$ iterations yields that the minimum $k$ for which
	The number of iterations $k$ is the minimum integer that satisfies the following inequality:
	
	\begin{equation*}
	\sum_{i=1}^k\left(\frac{3n}{2}-2i+1\right) \geq  \frac{n^2}{4}.
	\end{equation*}
	
	We have:
	
	\begin{equation*}
	\sum_{i=1}^k\left(\frac{3n}{2}-2i+1\right) = \left(\frac{3n}{2}+1\right)k - 2 \sum_{i=1}^k i = \left(\frac{3n}{2}+1\right)k - k^2 - k = \frac{3n}{2}k - k^2 
	\end{equation*}
	
	and, hence,
	
	\begin{equation*}
	k^2 - \frac{3n}{2}k + \frac{n^2}{4} \leq 0.
	\end{equation*}
	
	Equality is reached for 
	\begin{equation*}
	k_{1,2}=\frac{3n}{4} \pm \sqrt{\frac{9n^2}{16}-\frac{4n^2}{16}}=\frac{3\pm \sqrt{5}}{4}n.
	\end{equation*}
	
	This means that the inequality is satisfied for every value in the interval $I=[k_1,k_2]$.
	Since $k$ must be the minimum integer value in $I$, we have $k = \left\lceil \frac{3 - \sqrt{5}}{4}n \right\rceil$. 
	
	%We observe that $k_2=\frac{3 - \sqrt{5}}{4}n$ is the correct solution as $k_1=\frac{3 + \sqrt{5}}{4}n > \frac{n}{2}$, which cannot be the correct solution as there are only $\frac{n}{2}$ columns. Also note that for $k=\frac{n}{2} \in (k_2,k_1)$, we have 
	%\begin{equation*}
	%k^2 - \frac{3n}{2}k + \frac{n^2}{4} = \frac{n^2}{4} - \frac{3n^2}{4} + \frac{n^2}{4} = \frac{n^2}{4} < 0,
	%\end{equation*}
	%that is, the procedure ends only after at least $\left\lceil \frac{3 - \sqrt{5}}{4}n \right\rceil$ steps. 
\end{proof}

\section{Omitted Proofs of \cref{sec:pqn1}}
\label{app:pqn1}

\Cycle*
\label{lem:cycle*}

\begin{proof}
	For a leaf $\ell$, let $p(\ell)$ denote its neighbor in $C(L)$ and let $e^\star$ be an edge of $C(L)$ such that $w(e^\star)$ is maximum among all edges of $C(L)$.
	We remove $e^\star$ and every leaf~$\ell$ with $w({\ell}p(\ell))> w(e^\star)$ from $L$ obtaining a caterpillar $T$ where both degree-$1$ vertices of the underlying path $P(T)$ are endpoints of $e^\star$.
	We then find a PQ-layout of~$T$ using \cref{lem:caterpillar} chosing one of the endpoints of $e^\star$ as $r$.
	Note that for each edge $e$ in $T$ we have that $w(e^\star) \ge w(e)$.
	While~$e^\star$ may be co-active with edges in $e$, we have that its right endpoint succeeds any vertex in $T$, i.e., no forbidden configuration is created when we reinsert~$e^\star$.
	It remains to reinsert each leaf~$\ell$ with $w({\ell}p(\ell))> w(e^\star)$.
	Let $S$ denote the set of such leaves and for $\ell \in S$ let $w(\ell)=w({\ell}p(\ell))$.
	We insert all vertices of $S$ to the right of $r$ sorted in increasing order of weight.
	It remains to discuss that for $\ell \in S$, edge ${\ell}p(\ell)$ is not involved in forbidden configurations.
	Let $uv$ with $u \prec v$ be an co-active edge of ${\ell}p(\ell)$.
	If $v \not \in S$, we have that $v \prec \ell$ and $w(uv) \leq w(e^\star) <w({\ell}p(\ell))$, i.e., no forbidden configuration occurs.
	Otherwise, $v$ and $\ell$ are sorted by their weights in $\prec$ which are equal to $w(uv)$ and  $w({\ell}p(\ell))$, respectively.
	Again no forbidden configuration occurs and the statement follows.
	
	Finding and removing a maximum-weight edge of the cycle and all
	edges with greater weight can be done in~$\mathcal{O}(n)$ time.
	Arranging the resulting graph using \cref{lem:caterpillar} as well.
	For the subsequent reinsertion of edges,
	we need to sort the leaves by weight.
	This requires and can be done in $\mathcal{O}(n \log n)$ time,
	which is also the resulting running time of the algorithm.
\end{proof}

\TwoCaterpillars*
\label{lem:two-caterpillars*}

\begin{proof}
    First, consider the case that we have a 3-cycle $\triangle = abc$.
    Suppose that $w(ac) > \max\{w(ab),w(bc)\}$.
    Then, arrange the vertices of~$\triangle$ in order $a \prec b \prec c$.
    Clearly, this does not lead to any forbidden configuration of~$\triangle$.
    Lay out $C_a$ using \cref{lem:caterpillar} such that $a$ is the rightmost vertex of~$C_a$.
    Lay out $C_c$ using \cref{lem:tree} such that $c$ is the leftmost vertex of~$C_c$.
    Combine the PQ-layouts of $\triangle$, $C_a$, and $C_c$ such
    that all vertices of~$C_a$ precede all vertices of $\triangle$ (except for~$a$) and
    that all vertices of~$C_c$ succeed all vertices of $\triangle$ (except for~$c$).
    Clearly, this is a valid layout with priority queue number 1
    because the edge of $C_c$ and $\triangle$ are not co-active.
    
    Otherwise, there is an edge in~$\triangle$ that is at least as heavy as~$ac$.
    Without loss of generality, let $w(ab) \geq \max\{w(ac),w(bc)\}$.
    We handle this case as if we had a 4-cycle $\Box$ where we can safely ignore vertex~$d$
    (having the edge~$ac$ instead of a path via the edges~$ad$ and $cd$).
    It is easy to see that the rest of the proof applies to this case.
    
    Second, consider the case that we have a 4-cycle $\Box = abcd$.
	Assume, without loss of generality, that $w(ab)\geq \max\{w(ad),w(bc),w(cd)\}$.
	Consider the edge $e_c$ incident to~$c$ on $P(C_c)$.
	Assume first, that $w(e_c) \leq w(ab)$.
	For $\Box$, we use the layout $b \prec c \prec d \prec a$.
	Note that only the edge $ab$ nests and pseudo-nests edges of $\Box$, however it is the heaviest edge in $\Box$.
	Further, we use \cref{lem:tree} for laying out $C_a$ with leftmost vertex~$a$.
	The resulting PQ-layout of $G$ without $C_c$ clearly has priority queue number~1.
	Now, consider the caterpillar $C_c'$ obtained from $C_c$ by removing the leaves of~$c$.
	We lay out $C_c'$ using \cref{lem:caterpillar} and place all vertices of $C_c'$ except for $c$ before vertex $b$ in~$\prec$.
	Thus, there is one intersection between $C_c'$ and $\Box$, namely, between $e_c$ and~$ab$.
	We have that $c \prec a$ and $w(e_c) \leq w(ab)$, thus this creates no forbidden configuration. 
	
	Next, assume that $w(e_c) > w(ab)$.
	For $\Box$, we use the layout $a \prec d \prec c \prec b$.
	Again, only edge $ab$ nests and pseudo-nests edges of $\Box$, but still it is the heaviest edge in $\Box$.
	Further, we use  \cref{lem:caterpillar} for laying out $C_a$ with rightmost vertex $a$.
	Once more, the resulting PQ-layout of $G$ without $C_c$ clearly has priority queue number~1.
	Again, consider the caterpillar $C_c'$ obtained from $C_c$ by removing the leaves of $c$.
	We lay out $C_c'$ using \cref{lem:tree} and place all vertices of $C_c'$ except for $c$ after vertex $b$ in $\prec$.
	Thus, there is one intersection between $C_c'$ and~$\Box$, namely, between $e_c$ and $ab$, and one pseudo-nesting, namely between $e_c$ and $bc$.
	However, we have that $w(e_c) > w(ab) \geq w(bc)$, thus this creates no forbidden configuration. 
	
	Finally, in both cases where we consider~$\Box$, we only have to reintroduce the leaves $L_c$ of vertex $c$ in $C_c$ to complete the layout.
	For $\ell \in L_c$, let $w(\ell)=w({\ell}c)$.
	We iteratively insert the leaves of $c$ with increasing weight.
	If $w(\ell)\leq w(ab)$, we place $\ell$ directly before $c$ in $\prec$.
	Then, edge ${\ell}c$ only creates a nesting with edge $ab$, but $ab$ has heavier weight, so no forbidden configuration is introduced.
	Otherwise, $w(\ell) > w(ab)$ and we find a position for $\ell$ as follows.
	Let $C_r$ be the caterpillar whose non-root vertices are drawn to the right of the vertices of $\Box$, i.e., in the first case $C_a$ and in the second case $C_c'$.
	Also let $v^\star$ be the leftmost vertex of the caterpillar $C_r$ with parent $p(v^\star)$ in $C_r$ for which $w(v^{\star}p(v^\star))>w(\ell)$.
	Then, we place $\ell$ directly before $v^\star$ in $\prec$.
	For a contradiction, assume that this creates a forbidden configuration which clearly must involve edge ${\ell}c$ and another edge $uv$ with $u \prec v$ as our procedure ensures that the edges between $c$ and its leaves which pseudo-nest are not creating forbidden configurations.
	Observe that $uv$ is not an edge of $\Box$ as all edge weights of $\Box$ are smaller than the weight of edge ${\ell}c$ and $\ell$ is placed to the right of all vertices of~$\Box$.
	Thus $u = p(v)$ and if $v \prec \ell$ and $w({\ell}c) < w(uv)$, we have a contradiction to our choice for the position of $\ell$.
	Thus, we must have $\ell \prec v$ and $w({\ell}c) > w(uv)$.
	By our choice of the position of $\ell$, we have that $v^\star \prec v$.
	Since also $u \prec \ell$, 
	it follows that $uv$ and $p(v^\star)v^\star$ are co-active, however, this is a contradiction to the fact that $C_r$ is drawn with a priority-queue-number-one layout.
	The proof follows.
\end{proof}

\KTwoThree*
\label{lem:k23*}

\begin{figure}
	\centering
	\begin{subfigure}[t]{0.48\textwidth}
		\centering 
		\includegraphics[page=1]{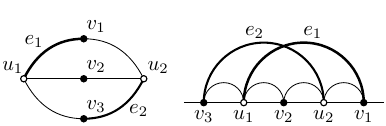}
		\subcaption{Case~1}
		\label{fi:k23-1}
	\end{subfigure}
	\hfill
	\begin{subfigure}[t]{0.48\textwidth}
		\centering 
		\includegraphics[page=2]{K23}
		\subcaption{Case~2.1}
		\label{fi:k23-2-1}
	\end{subfigure}
	\hfill
	\begin{subfigure}[t]{0.48\textwidth}
		\centering 
		\includegraphics[page=3]{K23}
		\subcaption{Case~2.2}
		\label{fi:k23-2-2}
	\end{subfigure}
	\hfill
	\begin{subfigure}[t]{0.48\textwidth}
		\centering 
		\includegraphics[page=4]{K23}
		\subcaption{Case~3}
		\label{fi:k23-3}
	\end{subfigure}
	\caption{Illustration for the proof of \cref{lem:k23}. The thickness of an edge indicates the weight of that edge.}\label{fi:k23}
\end{figure}

\begin{proof}
    Let $G=K_{2,3}$ and consult \cref{fi:k23} for an illustration.
    Let $\{u_1, u_2\}$ be the vertex partition with vertices of degree 3 and let $\{v_1, v_2, v_3\}$ be the vertex partition with vertices of degree 2.
    Let $w$ be any given edge-weight function and denote by $e_1$ the heaviest edge of $\langle G, w \rangle$ and by $e_2$ the second heaviest edge.
    We distinguish three main cases and, for each of them, we describe how to construct a PQ-layout $\Gamma$ such that $\pqn(\Gamma)=1$.
    
	\begin{itemize}
        \item {\bf Case~1: $e_1$ and $e_2$ are disjoint.} Refer to \cref{fi:k23-1}.
        We may assume that $e_1=u_1v_1$ and $e_2=u_2v_3$, and define the vertex ordering as $v_3 \prec u_1 \prec v_2 \prec u_2 \prec v_1$. Since $e_1$ and $e_2$ are the two heaviest edges, there are neither forbidden nestings nor forbidden pseudo-nestings in $\Gamma$. Also, the only two crossing edges are $e_1$ and $e_2$, and $w(e_2) \leq w(e_1)$. Since $e_2$ is pulled from the priority queue before $e_1$, there is no forbidden crossing configuration.
		\item {\bf Case~2: $e_1$ and $e_2$ share a degree-2 vertex.} Let $e_3$ be the third heaviest edge of $\langle G, w \rangle$. We consider two further subcases:
		\begin{itemize}
			\item {\bf Case~2.1: $e_3$ shares a vertex with $e_1$.} Refer to \cref{fi:k23-2-1}. Assume, w.l.o.g., that $e_1=u_1v_3$, $e_2=u_2v_3$, and $e_3=u_1v_1$. The linear ordering of $\Gamma$ is set as $v_1 \prec u_2 \prec v_2 \prec u_1 \prec v_3$. As for Case~1, there are no forbidden (pseudo-)nesting configurations. Also, the only two crossing edges are $e_2$ and $e_3$, and $w(e_3) \leq w(e_2)$. Since $e_3$ is pulled from the queue before $e_2$, there is no forbidden crossing~configuration.
			\item {\bf Case~2.2: $e_3$ shares a vertex with $e_2$.} Refer to \cref{fi:k23-2-2}.  Assume, w.l.o.g., that $e_1=u_1v_3$, $e_2=u_2v_3$, and $e_3=u_2v_1$. We define the linear ordering of $\Gamma$ as $v_1 \prec u_1 \prec v_2 \prec u_2 \prec v_3$. The argument is similar to the previous case.
		\end{itemize}
		\item {\bf Case~3: $e_1$ and $e_2$ share a degree-3 vertex.} Refer to \cref{fi:k23-3}. Assume, w.l.o.g., that $e_1=u_1v_3$ and $e_2=u_1v_1$. We define the linear ordering of $\Gamma$ as $v_3 \prec v_1 \prec u_2 \prec v_2 \prec u_1$. As in the previous cases, for the choice of $e_1$ and $e_2$, there are no forbidden nesting or pseudo-nesting configurations. Also, the only two crossing edges are $e_2$ and the edge $v_3u_2$; since $v_3u_2$ is pulled before $e_2$ from the priority queue, and since $w(v_3u_2) \le w(e_2)$, there is no forbidden crossing configuration. \qedhere
	\end{itemize}
\end{proof}

\KFourMinusE*
\label{lem:k4minusE*}

\begin{figure}
    \centering
    \begin{subfigure}[t]{0.48\textwidth}
        \centering 
        \includegraphics[page=1]{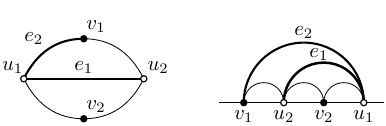}
        \subcaption{Case~1}
        \label{fi:k4minusE-1}
    \end{subfigure}
    
    \begin{subfigure}[t]{0.48\textwidth}
        \centering 
        \includegraphics[page=2]{K4minusE}
        \subcaption{Case~2.1}
        \label{fi:k4minusE-2-1}
    \end{subfigure}
    \hfill
    \begin{subfigure}[t]{0.48\textwidth}
        \centering 
        \includegraphics[page=3]{K4minusE}
        \subcaption{Case~2.2}
        \label{fi:k4minusE-2-2}
    \end{subfigure}
    
    \begin{subfigure}[t]{0.48\textwidth}
        \centering 
        \includegraphics[page=4]{K4minusE}
        \subcaption{Case~2.3}
        \label{fi:k4minusE-2-3}
    \end{subfigure}
    \hfill
    \begin{subfigure}[t]{0.48\textwidth}
        \centering 
        \includegraphics[page=5]{K4minusE}
        \subcaption{Case~2.4.1}
        \label{fi:k4minusE-2-4-1}
    \end{subfigure}
    
    \begin{subfigure}[t]{0.48\textwidth}
        \centering 
        \includegraphics[page=6]{K4minusE}
        \subcaption{Case~2.4.2}
        \label{fi:k4minusE-2-4-2}
    \end{subfigure}
    \hfill
    \begin{subfigure}[t]{0.48\textwidth}
        \centering 
        \includegraphics[page=7]{K4minusE}
        \subcaption{Case~2.4.3}
        \label{fi:k4minusE-2-4-3}
    \end{subfigure}
    \caption{Illustration for the proof of \cref{lem:k4minusE}. The thickness of an edge indicates the weight of that edge.}\label{fi:k4minusE}
\end{figure}

\begin{proof}
    Let $G$ be $K_4$ minus an edge and consult \cref{fi:k4minusE} for an illustration.
    Let $\{u_1, u_2\}$ be the vertices of degree 3 and let $\{v_1, v_2\}$ be the vertices of degree 2.
    Let $w$ be any given edge-weight function and denote by $e_1$ the heaviest edge of $\langle G, w \rangle$,
    by $e_2$ the second heaviest edge, and by $e_3$ the third heaviest edge.
    We distinguish two main cases and, for the second case, four subcases.
    We describe how to construct a PQ-layout $\Gamma$ such that $\pqn(\Gamma)=1$.
    
    \begin{itemize}
        \item {\bf Case~1: $e_1 = u_1u_2$.} Refer to \cref{fi:k4minusE-1}.
        Assume, without loss of generality (due to symmetry), that $e_2 = u_1v_1$.
        We define the vertex ordering as $v_1 \prec u_2 \prec v_2 \prec u_1$.
        The only edges spanning over other edges are~$e_1$ and~$e_2$.
        Since they are the heaviest edges and end at the rightmost vertex, there is no forbidden configuration with another edge.
        Furthermore, $e_1$ and $e_2$ together do not create a forbidden configuration since they share a common right endpoint.
        \item {\bf Case~2: $e_1 = u_1v_1$.} Note that this is the only remaining case to assign $e_1$ since
        all edges except for $u_1u_2$ are equivalent by the symmetry of~$G$.
        We consider four subcases for the four choices of the edge~$e_2$,
        and, in the last subcase, we distinguish further between the choices for edge~$e_3$.
        \begin{itemize}
            \item {\bf Case~2.1: $e_2 = u_1u_2$.} Refer to \cref{fi:k4minusE-2-1}.
            As in Case~1, set the linear ordering of the vertices to $v_1 \prec u_2 \prec v_2 \prec u_1$.
            Again, the only edges spanning over other edges are~$e_1$ and~$e_2$,
            and $e_1$ and~$e_2$ share a common right endpoint.
            \item {\bf Case~2.2: $e_2 = u_2v_2$.} Refer to \cref{fi:k4minusE-2-2}.
            Set the linear ordering of the vertices to $v_2 \prec u_1 \prec u_2 \prec v_1$.
            Again, the only edges spanning over other edges are~$e_1$ and~$e_2$.
            Now, however, $e_1$ and $e_2$ intersect, but they do not create a forbidden configuration
            because $w(e_1) > w(e_2)$ and $e_1$ ends after $e_2$.
            \item {\bf Case~2.3: $e_2 = u_1v_2$.} Refer to \cref{fi:k4minusE-2-3}.
            Set the linear ordering of the vertices to $v_1 \prec v_2 \prec u_2 \prec u_1$.
            Since the edges $e_1$ and $e_2$ are the heaviest edges,
            end at the rightmost vertex, and share a common right endpoint,
            they are not involved in any forbidden configuration.
            The only remaining pair of co-active edges is $u_2v_1$ and $u_2v_2$.
            Since they share $u_2$ as their common right endpoint,
            they do not induce a forbidden configuration.
            \item {\bf Case~2.4.1: $e_2 = u_2v_1$ and $e_3 = u_1u_2$.} Refer to \cref{fi:k4minusE-2-4-1}.
            Set the linear ordering of the vertices to $u_1 \prec v_2 \prec u_2 \prec v_1$.
            Note that $e_1$ spans over all edges but also has the largest weight,
            so $e_1$ is not involved in a forbidden configuration.
            The only remaining edge that spans over other edges is~$e_3$.
            However, note that both $u_1v_2$ and $u_2v_2$ do not have a greater weight than~$e_3$,
            so there is also no forbidden configuration.
            \item {\bf Case~2.4.2: $e_2 = u_2v_1$ and $e_3 = u_1v_2$.} Refer to \cref{fi:k4minusE-2-4-2}.
            Set the linear ordering of the vertices to $u_1 \prec u_2 \prec v_2 \prec v_1$.
            The two heaviest edges ($e_1$ and~$e_2$) have the rightmost endpoint,
            so they are not involved in a forbidden configuration.
            No edge except for $e_1$ and $e_2$ end to the right of
            the third heaviest edge $e_3$, so $e_3$ is not involved in a forbidden configuration.
            The other two remaining edges are not co-active.
            \item {\bf Case~2.4.3: $e_2 = u_2v_1$ and $e_3 = u_2v_2$.} Refer to \cref{fi:k4minusE-2-4-3}.
            Set the linear ordering of the vertices to $u_2 \prec u_1 \prec v_2 \prec v_1$.
            There are no forbidden configurations by the same argument as in Case~2.4.2;
            the only difference is that~$e_1$ and~$e_2$ are swapped, but both end at the same right endpoint. \qedhere
        \end{itemize}
    \end{itemize}
\end{proof}

\LastEdgeHeavy*
\label{le:1pq-cycle-last-vertex*}

\begin{proof}
    Suppose for a contradiction that there is a PQ-layout~$\Gamma$ of $\langle C, w \rangle$ with $\pqn(\Gamma) = 1$ where the last vertex is incident to two edges that both do not have the maximum weight.
    Let $e = uv$ be an edge of maximum weight with $u \prec v$.
    (If there are multiple edges with maximum weight,
    pick $e$ such that $v$ is rightmost among all right endpoints
    of those heaviest edges.)
    We claim that there is an edge $e' = xy$ that is co-active with $e$
    and $v \prec y$ resulting in a forbidden configuration since $w(e) > w(e')$.
    Let $a$ be the leftmost vertex and $z$ be rightmost vertex of~$\Gamma$.
    There are two vertex-disjoint paths between $a$ and $z$.
    Only one of these paths can contain $v$, while the other needs to bridge $v$~-- this implies that $e'$ exists.
\end{proof}

\ForbiddenMinors*
\label{le:forbidden-minors-are-not-1pq*}

\begin{proof}[Proof of Cases $\mathcal{F}_4$--$\mathcal{F}_8$]
It remains to analyze $\mathcal{F}_4$--$\mathcal{F}_8$.
For an illustration, refer to \cref{fig:minors-copy}.
	
\begin{description}
    \item[$\mathcal{F}_4$]:
    We consider the cases that $G = \mathcal{F}_4$
    and that $G$ is larger graph that has $\mathcal{F}_4$ as a minor both at once.
    With a sequence of vertex splits (potentially after renaming), we can
    (i)~replace the edge~$bc$ by a path (we denote a sequence of additional vertices by $x^*$), or
    (ii)~replace the edge~$ad$ by a path (we denote a sequence of additional vertices by $y^*$).
    In any case, we assign any additional edge a weight of~2.
    
    \begin{figure}
        \centering
        \begin{subfigure}[t]{0.1125\textwidth}
            \centering 
            \includegraphics[page=10]{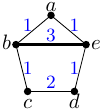}
            \subcaption{$\mathcal{F}_1$}
            \label{fig:minors:1-copy}
        \end{subfigure}
        \hfill
        \begin{subfigure}[t]{0.1225\textwidth}
            \centering 
            \includegraphics[page=11]{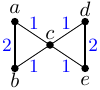}
            \subcaption{$\mathcal{F}_2$}
            \label{fig:minors:2-copy}
        \end{subfigure}
        \hfill
        \begin{subfigure}[t]{0.125\textwidth}
            \centering 
            \includegraphics[page=12]{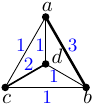}
            \subcaption{$\mathcal{F}_3$}
            \label{fig:minors:3-copy}
        \end{subfigure}
        \hfill
        \begin{subfigure}[t]{0.1125\textwidth}
            \centering 
            \includegraphics[page=13]{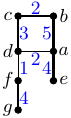}
            \subcaption{$\mathcal{F}_4$}
            \label{fig:minors:4-copy}
        \end{subfigure}
        \hfill
        \begin{subfigure}[t]{0.1225\textwidth}
            \centering 
            \includegraphics[page=14]{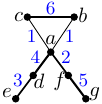}
            \subcaption{$\mathcal{F}_5$}
            \label{fig:minors:5-copy}
        \end{subfigure}
        \hfill
        \begin{subfigure}[t]{0.145\textwidth}
            \centering 
            \includegraphics[page=15]{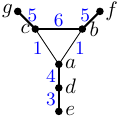}
            \subcaption{$\mathcal{F}_6$}
            \label{fig:minors:6-copy}
        \end{subfigure}
        \hfill
        \begin{subfigure}[t]{0.1025\textwidth}
            \centering 
            \includegraphics[page=16]{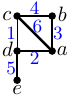}
            \subcaption{$\mathcal{F}_7$}
            \label{fig:minors:7-copy}
        \end{subfigure}
        \hfill
        \begin{subfigure}[t]{0.1025\textwidth}
            \centering 
            \includegraphics[page=17]{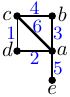}
            \subcaption{$\mathcal{F}_8$}
            \label{fig:minors:8-copy}
        \end{subfigure}
        \caption{(Copy of \cref{fig:minors}.) Forbidden minors for graphs with priority queue number~1.
            The graphs with the provided edge weights do not admit a linear layout with priority queue number~1.
            The thickness of an edge indicates the weight of that edge.}
        \label{fig:minors-copy}
    \end{figure}
    
    By \cref{le:1pq-cycle-last-vertex},
    the last vertex of the cycle $abx^*cdy^*$ is $a$ or~$b$.
    If~$a$ is the last vertex, then also $c \prec d$,
    as otherwise $cd$ ends below a lighter edge of weight~2 on the path from~$d$ to~$a$.
    Furthermore, $b \prec d$, as otherwise $cd$ ends below
    a lighter edge of weight~2 on the path between~$b$ and~$d$.
    The symmetric arguments apply when we swap the roles of~$a$ and~$b$
    and simultaneously the roles of~$c$ and~$d$.
    
    \smallskip
    
    This leaves four possible orderings:
    If $a$ is the last vertex, we have
    $b \prec c \prec d \prec a$ or $c \prec b \prec d \prec a$.
    If $b$ is the last vertex, we have
    $a \prec d \prec c \prec b$ or $d \prec a \prec c \prec b$.
    
    \smallskip
    
    First, suppose that $b \prec c \prec d \prec a$.
    No vertex from~$x^*$ on the path between $b$ and $c$
    lies to the right of $d$ since then an edge of weight~2
    spans over the end of the heavier edge~$cd$ of weight~3.
    Observe that $c \prec f \prec a$ because $df$ has the lowest weight
    and there are heavier edges ending at~$a$ and at $c$ and~$x^*$.
    Now $f$ cannot be the right endpoint of $fg$
    because then the edge $fg$ would end below
    the lighter edge~$cd$ or below a lighter edge on the path from~$d$ to~$a$.
    However, $f$ can also not be the left endpoint of~$fg$
    because then, if $f \prec g \prec a$, the edge~$fg$ would end below the lighter edge~$cd$
    or below a lighter edge on the path from~$d$ to~$a$,
    or, if $f \prec a \prec g$, the heavier edge $ab$ would end below~$fg$.
    
    \smallskip
    
    Second, suppose that $c \prec b \prec d \prec a$.
    Then again no vertex from~$x^*$ on the path between $b$ and $c$
    lies to the right of $d$.
    By swapping the roles of $b$ and~$c$,
    we can use the same argument as before to obtain~$b \prec f \prec a$.
    Again, $f$ cannot be the right endpoint of~$fg$
    because then $fg$ would end below
    the lighter edge~$cd$ or below a lighter edge on the path from~$d$ to~$a$,
    and it cannot be the left endpoint of $fg$ by
    the same argument as before.
    
    \smallskip
    
    Third, suppose that $a \prec d \prec c \prec b$.
    We next locate~$f$.
    Since $c$ is the right endpoint of the edge $cd$ with weight~3,
    it follows that $f \prec c$ as the edge $df$ has weight~1.
    Suppose that $a \prec f \prec c$.
    As before, $f$ cannot be the right endpoint of $fg$
    because then the edge $fg$ would end below
    a lighter edge on the path from~$a$ to~$d$ or below the lighter edge~$cd$.
    Also, $f$ cannot be the left endpoint of~$fg$
    because then, if $f \prec g \prec b$, the edge~$fg$ would end below the lighter edge~$cd$
    or below a lighter edge on the path from~$c$ to~$b$,
    or, if $f \prec b \prec g$, the heavier edge $ab$ would end below~$fg$.
    So, suppose that $f \prec a$.
    It follows that $a \prec e$ and $g \prec f$.
    It remains to give~$e$ a precise position.
    If $a \prec e \prec b$, then $ae$ ends
    below one of the lighter edges on the path from~$a$ to~$d$, or below $cd$, or
    below one of the lighter edges on the path from~$c$ to~$d$.
    However, if $b \prec e$, then the heavier edge $ab$ ends below~$ae$.
    
    \smallskip
    
    Finally, suppose that $d \prec a \prec c \prec b$.
    Observe that we can use similar arguments as in the previous case
    to show that $g \prec f \prec d \prec a \prec c \prec b$.
    It remains to place~$e$.
    Clearly, $a \prec e$ as otherwise $ae$ ends below~$cd$.
    Now, however, $ae$ ends below the lighter edge~$cd$ or a lighter edge on the path from~$c$ to~$b$,
    or the heavier edge $ab$ ends below $ae$;
    a contradiction.
    
    \todo[inline]{Do missing cases and increase index of every graph by 1! (If it gets too complicated, don't handle two cases at once as it is done currently)}
    
	\item[$\mathcal{F}_5$, $\mathcal{F}_6$]:
	We first consider the case $G\in \{\mathcal{F}_5,\mathcal{F}_6\}$. By \cref{le:1pq-cycle-last-vertex} and symmetry, we may assume that $a,b \prec c$.		
	Since each edge \emph{not} in $\triangle abc$ has a weight strictly between $w(ab) = w(ac)$ and $w(bc)$, no such edge can start weakly left of $\min\{a,b\}$ and end strictly right of $\min\{a,b\}$.
	In particular, it follows that $d \prec a$ and, since also $w(ed) < w(ad)$, that $e \prec d$.
	In summary, we have $e \prec d \prec a \prec c$ and $b \prec c$.
	
	Now for $\mathcal{F}_5$ it follows that the right endpoint of edge $fg$ must be left of $e$.
	But then edge $ed$ nests below $af$; a contradiction, as $w(ed) > w(af)$.
	
	For $\mathcal{F}_6$, first note that if $a \prec b$, then either edge $bf$ ends below $ac$ (if $f \prec c$) or edge $bc$ ends below $bf$ (if $c \prec f$); either is not allowed.
	So we have $b \prec a$, and hence $f \prec a$, as argued in the beginning.
	But now either edge $bf$ ends below $ad$ (if $d \prec b$) or edge $ed$ ends below $ab$ (if $b \prec d$); either is not allowed.
	
	Consider now the case where $G$ contains an $\mathcal{F}_5$ minor, i.e., $G$ contains a subgraph isomorphic to a subdivision of $\mathcal{F}_5$.  That is, $G$ consists of a cycle $C$ that contains a vertex $a$ and a non-adjacent edge $bc$, a simple path $P_a$  from $a$ to a vertex $a'$ that contains no edge of $C$ and two paths $P_{a'e}=a'de$ and $P_{a'g}=a'fg$ distinct from $C$ and $P_a$. We construct the weight function $w$ similar as before, namely, we set $w(bc)=6$, otherwise $w(e)=0$ for each edge $e$ on $C$, $w(e_a)=1$ for each edge $e_a$ on $P_a$, and $w(a'd)=4$, $w(de)=3$, $w(a'f)=2$, $w(fg)=5$. Once more, by \cref{le:1pq-cycle-last-vertex} and symmetry, we may assume that $a,b \prec c$ as $bc$ is the heaviest edge on $C$.		
	
	We next claim that $a' \prec a$. First observe that $'a \prec c$ as otherwise edge $bc$ ends at $c$ but there is at least one edge $e_a$ of the path $P_a$ being co-active with $bc$ that ends after $c$ which is a contradiction to the fact that $w(e_a) < w(bc)$. Assume now for a contradiction that $a \prec a' \prec c$ and consider the simple path $P_{ac}$ between $a$ and $c$ along $C$ whose edges all have weight $0$. Since $P_a$ and $P_{ac}$ both start at $a$ and $a' prec c$ there must be an edge $e_a$ of $P_a$ and an edge $e_{ac}$ of $P_{ac}$ that are co-active and such that $e_a$ ends before $e_ac$. This is however a contradiction to the fact that $w(e_{ac}) < w(e_a)$.
	
	Hence, we can assume from now on that $a' \prec a$. As for the case of $\mathcal{F}_5$, we can now observe that all of edges $a'd$, $de$, $a'f$ and $fg$ have higher weights then any of the already drawn edges except for $bc$. Thus, none of these edges can start weakly left of $\min\{a',b\}$ and end strictly right of $\min\{a',b\}$. Hence, $e,d,f,g \prec a'$ and $e \prec d$. As for $\mathcal{F}_5$, the right endpoint of $fg$ must be left of $e$ which leads to $ed$ nesting $a'f$; a contradiction.

	Finally, consider the case where $G$ contains an $\mathcal{F}_6$ minor, i.e., it contains a subdivision of $\mathcal{F}_6$ as a minor. In this scenario, we observe that one of the following two scenarios happens: If we only subdivide any of the edges $ad$, $de$, $bf$ or $cg$, we obtain a graph that again contains $\mathcal{F}_6$ as a subgraph. Otherwise, we subdivide the triangle $\triangle abc$. However, subdividing the triangle $\triangle abc$ once at a single edge provides a supergraph of $\mathcal{F}_4$. Thus, any subdivision leads either to an $\mathcal{F}_6$ subgraph or to an $\mathcal{F}_4$ minor which concludes the case.
	
	\item[$\mathcal{F}_7$, $\mathcal{F}_8$]: We first consider the case where $G \in \{\mathcal{F}_7, \mathcal{F}_8\}$.
	By \cref{le:1pq-cycle-last-vertex},
	$b$ or $c$ is the last vertex of the 4-cycle $abcd$
	and $a$ or $c$ is the last vertex of $\triangle abc$.
	Hence, $a, b, d \prec c$.
	Moreover, $a \prec d$ since otherwise, $ad$ would end below the lighter edge $cd$.
	We know that $e \prec c$ because otherwise $ac$
	would end below the lighter edge $de$ ($ae$, respectively).
	
	For $\mathcal{F}_7$, further $e \prec d$ because
	if $d \prec e \prec c$, then $de$ would end below the lighter edge~$cd$.
	It remains to place $b$ (we already know $b \prec c$).
	It cannot be $b \prec d$ because then
	$de$ would end below the lighter edge $bc$.
	Hence, we have $a \prec d \prec b \prec c$.
	Now, however, $ab$ ends below the lighter edge $cd$; a contradiction.
	
	For $\mathcal{F}_8$, a similar argument applies.
	We have $e \prec a$ because
	if $a \prec e \prec c$, then $ae$ would end below
	one of the lighter edges~$ad$ or~$cd$.
	If $b \prec a \prec d \prec c$,
	then $ae$ would end below the lighter edge $bc$.
	If $a \prec b \prec d \prec c$,
	then $ab$ would end below the lighter edge $ad$.
	Hence, we have $a \prec d \prec b \prec c$.
	Now, however, $ab$ ends below the lighter edge $cd$; a contradiction.
	
	Finally, consider the case where $G$ has a proper $\mathcal{F}_7$ or $\mathcal{F}_8$ minor, i.e., it contains a subgraph isomorphic to subdivision of $\mathcal{F}_7$ or $\mathcal{F}_8$, respectively. If the subdivision is obtained only by subdividing edge $de$ or $ae$, the obtained graph contains $\mathcal{F}_7$ or $\mathcal{F}_8$, respectively. Otherwise, if any of the edges $ab$, $bc$, $cd$ or $da$ are subdivided, the graph contains a subdivision of $\mathcal{F}_1$. 
	
	Hence, we may assume that the only subdivided edge is $ac$. Moreover, we can assume that $ac$ is subdivided exactly once as otherwise we again yield a graph that contains a $\mathcal{F}_1$ minor. Let $f$ denote the subdivision vertex of $ac$ and let $w(af)$ and $w(cf)=2$.

	By \cref{le:1pq-cycle-last-vertex},
	$b$ or $c$ is the last vertex of the $4$-cycle $abcd$
	and $a$ or $f$ is the last vertex of the $4$-cycle $abcf$. It follows, that $a,b,c,d \prec f$.
	Next, we observe that $b \prec c$ as otherwise the edges $cb$ and $cf$ start at $c$, i.e., they are co-active, but $bc$ ends before $cf$ despite $w(bc) > w(cf)$. Hence, $c$ is the last vertex of the $4$-cycle $abcd$ and as previously we obtain $a, b, d \prec c$. 	Moreover, as before $a \prec d$ since otherwise, $ad$ would end below the lighter edge co-active edge $cd$.

	Up to this point, we established that $a \prec d \prec c \prec f$. We now place vertex $b$. If $a \prec b \prec c$, edge $ab$ would be co-active with and end before either $ad$ or $dc$ despite the fact that $w(ab) > w(ad) > w(dc)$. On the other hand, if $c \prec b \prec f$, edges $bc$ and $cf$ are co-active and the heavier edge $bc$ ends before the lighter edge $cf$. Finally, if $f \prec b$, the heavier edge $af$ ends before the lighter co-active edge $bc$. Hence, we conclude that $b \prec a \prec d \prec c \prec f$.
	
	Now consider the placement of $e$. If $e \prec c$, the edge $de$ or $ae$ (in the subdivision of $\mathcal{F}_7$ or $\mathcal{F}_8$, respectively) is co-active with $bc$ but ends before it even though it is heaver than $bc$. If $c \prec e \prec f$, edge $de$ or $ae$  is co-active with the lighter edge $cf$ but ends before it. Finally, if $f \prec e$, edge $de$ or $ae$ is co-active with the heavier edge $af$ but ends after it. Thus, there is no valid placement for $e$ and $pqn(G)>1$.
	\qedhere
\end{description}	
\end{proof}

\pqOne*
\label{thm:pq1*}

\begin{proof}
    Let $G = (V,E)$ be any fixed graph.
    Our task is to show that either $\pqn(G) = 1$, or $G$ contains some $\mathcal{F} \in \{\mathcal{F}_1,\ldots,\mathcal{F}_8\}$ as a minor.
    For this, we may assume that $G$ is connected.
    If $G$ is a tree, i.e., contains no cycles, then $\pqn(G) = 1$ by \cref{lem:tree} and we are done.
    
    Next assume that $G$ contains exactly one cycle $C$.
    Then $G - E(C)$ is a forest $F$, and we consider each component of $F$ as rooted at the vertex of $C$.
    For a component $K$ of $F = G - E(C)$ we say that $K$ is
    \begin{itemize}
        \item \emph{trivial} if $K$ is just an isolated vertex,
        \item a \emph{rooted star} if $K$ is a star whose central vertex lies on $C$,
        \item a \emph{rooted caterpillar} if $K$ is a caterpillar with an underlying path with one endpoint on $C$,
        \item \emph{complex} if none of the above applies for $K$.
    \end{itemize}
    %We classify each component of~$F$ as trivial (being an isolated vertex), a rooted star, a caterpillar rooted at a degree-$1$ vertex of an underlying path (for one choice of an underlying path), or anything more complex.
    First assume that some component $K$ of $F$ is complex.
    Let $w$ be the vertex in $K$ that lies on $C$.
    If $K$ is not a caterpillar, then $K$ contains a vertex~$v$ (possibly $w$) with at least three neighbors of degree at least~$2$ each.
    Contracting $C$ to a triangle and the $v$-$w$-path in $K$ to a single vertex results in a graph that contains $\mathcal{F}_5$.
    Hence, $G$ contains $\mathcal{F}_5$ as a minor.
    So assume that $K$ is a caterpillar but no underlying path of $K$ starts with $w$.
    Let $P$ be the path in $K$ on all non-leaf vertices of $K$.
    Then $w$ is neither an endpoint of $P$ nor adjacent to an endpoint of $P$.
    It follows that either $w$ is an interior vertex of $P$ or $w$ is adjacent to an interior vertex $w'$ of $P$.
    Contracting $C$ again to a triangle, and contracting the edge $ww'$ in the latter case, we again obtain a graph that contains $\mathcal{F}_5$.
    Hence again, $G$ contains $\mathcal{F}_5$ as a minor and we are left with the case that no component of $F$ is complex.

    Now if every component of $F$ is trivial or a rooted star, then $\pqn(G) = 1$ by \cref{lem:cycle}.
    So assume some component $K$ of $F$ is a rooted caterpillar.
    If all other components (different from~$K$) are trivial, then $\pqn(G)=1$ by \cref{lem:one-caterpillar}.
    If $F$ has besides $K$ two further non-trivial components $K',K'' \neq K$, then $G$ contains $\mathcal{F}_6$ as a minor.
    So assume now that there is exactly one further non-trivial component~$K' \neq K$.
    If the roots of $K$ and $K'$ have distance at least $3$ in one direction along $C$, then $G$ contains $\mathcal{F}_4$ as a minor.
    So, if none of the above is the case, then either $C$ is a triangle and $\pqn(G)=1$ by \cref{lem:two-caterpillars}, or $C$ is a quadrangle and the roots of $K$ and $K'$ are opposite on $C$, and $\pqn(G) = 1$ by \cref{lem:two-caterpillars}.
    
    The case where~$G$ contains at least two cycles remains.
    First assume that $G$ contains two edge-disjoint cycles $C_1,C_2$.
    If $C_1,C_2$ have at most one vertex in common,
    then $G$ contains $\mathcal{F}_2$ as a minor.
    If $C_1,C_2$ have at least two vertices in common, we claim that $G$ contains $\mathcal{F}_8$ as a minor.
    One of $C_1,C_2$, say $C_2$, has length at least~$4$.
    Let $a$ and $c$ be two common vertices of $C_1,C_2$ (consider \cref{fig:minors:7} for the naming scheme).
    Then $C_1$ can be contracted to the triangle $abc$,
    while $C_2$ and can be contracted to the 4-cycle $adce$.
    Omitting edge~$ce$ yields~$\mathcal{F}_8$.
    
    Finally, consider the case that any two cycles in $G$ share at least one edge.
    In fact, either $G$ contains $\mathcal{F}_3 = K_4$ as a minor (and we are done) or there are two vertices $x,y$ in $G$ that are both contained in all cycles.
    Since there are no edge-disjoint cycles, there are at most three edge-disjoint $x$-$y$-paths in $G$.
    As $G$ has two cycles, there are exactly three such $x$--$y$-paths, say $P_1,P_2,P_3$.
    If $G$ is not just $P_1\cup P_2\cup P_3$, then $G$ contains $\mathcal{F}_7$ or $\mathcal{F}_8$ as a minor.
    If at least one of $P_1,P_2,P_3$ has at least three edges, then $G$ contains $\mathcal{F}_1$ as a minor.
    So, if none of the above is the case, then either $G = K_{2,3}$ or $G = K_4-e$ and we have $\pqn(G) = 1$ by \cref{lem:k23,lem:k4minusE}.
\end{proof}

\section{Omitted Content of \cref{sec:results-kpq}}
\label{app:pathwidth-treewidth}
Here, we provide a formal definition of treewidth and pathwidth.
%We start with treewidth and define pathwidth afterwards.
%
Given an undirected graph $G = (V, E)$, a \emph{tree decomposition}
is a tree $T$ with nodes $X_1, \dots, X_k$ such that:
\begin{itemize}
	\item for each $i \in \{1, \dots, k\}$, $X_i \subseteq V$,
	\item $\bigcup_{i \in \{1, \dots, k\}} X_i = V$,
	\item if $X_i$ and $X_j$ both contain a vertex $v \in V$,
	then all nodes of $T$ that lie on the (unique) path
	between $X_i$ and $X_j$ contain $v$ as well, and
	\item for every $uv \in E$, there is some $i \in \{1, \dots, k\}$ such that $\{u, v\} \subseteq X_i$.
\end{itemize}
The \emph{width} of a tree decomposition is $\max\{|X_i| - 1\ : i \in \{1, \dots, k\}\}$, i.e., it is the size of its largest set $X_i$ minus one.
The \emph{treewidth} of $G$ is the minimum width among all possible tree decompositions of~$G$.
The \emph{pathwidth} is defined similarly.
The only difference is that $T$ needs to be a path,
which is hence known as \emph{path decomposition}.

Now, we present the full proof of \cref{thm:treewidth-2}.

\TreewidthTwo*
\label{thm:treewidth-2*}
\begin{proof}
    Let $p \geq 1$ be a fixed integer.
    Our desired graph $G_p$ will be a \emph{rooted $2$-tree}, i.e., an edge-maximal graph of treewidth~$2$ with designated root edge, which are defined inductively through the following construction sequence:
    \begin{itemize}
        \item A single edge $e^* = v'v''$ is a $2$-tree rooted at $e^*$.
        \item If $G$ is a $2$-tree rooted at $e^*$ and $e = uv$ is any edge of $G$, then adding a new vertex $t$ with neighbors $u$ and $v$ is again a $2$-tree rooted at $e^*$.
            In this case we say that $t$ is \emph{stacked onto the edge $uv$}, and vertices $u$ and $v$ are called the \emph{parents} of $t$.
    \end{itemize} 
    
    We define a vertex ordering $\sigma$ of a $2$-tree rooted at $e^* = v'v''$ to be \emph{left-growing} if no vertex (different from $v'$, $v''$) appears to the right of both its parents.
    
    \LeftGrowing*
    \label{claim:wlog-left-growing*}
    
    \begin{claimproof}
        Given $\langle G,w \rangle$, let us fix some $\varepsilon > 0$ smaller than the minimum difference between any two distinct edge-weights in $\langle G,w \rangle$, i.e., $0 < \varepsilon < \min\{ |w(e_1)-w(e_2)| \colon e_1,e_2 \in E(G), w(e_1) \neq w(e_2)\}$.
        Let us call a weight function $w' \colon E(G) \to \mathbb{R}$ to be \emph{$\varepsilon$-close} to $w$ if $|w'(e) - w(e)| \leq \varepsilon/2$ for every edge $e \in E(G)$.
        Clearly, for any fixed vertex ordering, a $p$-inversion for edge weights~$w$ is also a $p$-inversion for $w'$, if $w'$ is $\varepsilon$-close to $w$.
        
        Now consider the given $\langle G,w \rangle$ and a construction sequence of $G$ starting with the root edge $e^*$.
        We shall construct $2$-tree $\tilde{G}$ such that for each edge $e$ in $G$ there is a corresponding set of edges in $\tilde{G}$, which we call \emph{copies} of $e$.
        It starts with exactly $1 = p^0 = p^{2 \cdot 0}$ copies of the root edge $e^*$.
        Then, whenever a vertex $t$ is stacked onto an edge $e = uv$ in the construction sequence of $G$, for each existing copy $\tilde{e} = \tilde{u}\tilde{v}$ of $e$ we stack $p^2$ vertices $t_1,\ldots,t_{p^2}$ on $\tilde{e}$ in $\tilde{G}$.
        As edge-weights we choose
        \[
            \tilde{w}(\tilde{u}t_i) = w(ut) + i \cdot \frac{\varepsilon}{2p^2} \qquad \text{and} \qquad \tilde{w}(\tilde{v}t_i) = w(vt) - i \cdot \frac{\varepsilon}{2p^2} \qquad \text{for $i=1,\ldots,p^2$.}
        \]
        Each such edge $\tilde{u}t_i$ is a copy of $ut$, while each such $\tilde{v}t_i$ is a copy of $vt$.
        Thus, if there are $p^{2j}$ copies of $e$ in $\tilde{G}$, then there are $p^{2(j+1)}$ copies of $ut$ and $p^{2(j+1)}$ copies of $vt$ in $\tilde{G}$.
        This completes the construction of $\langle \tilde{G},\tilde{w} \rangle$.
        
        Now fix any vertex ordering $\sigma$ of $\tilde{G}$.
        Consider a vertex $t$ that is stacked onto an edge $e = uv$ of $G$, and one copy $\tilde{e}  = \tilde{u}\tilde{v}$ of $e$ in $\tilde{G}$.
        If all $p^2$ corresponding vertices $t_1,\ldots,t_{p^2}$ in $\tilde{G}$ lie right of both vertices of $\tilde{e}$ in $\sigma$, then by the Erd\H{o}s--Szekeres theorem there is a subset of at least $p$ such vertices with all decreasing or all increasing indices in the left-to-right order in $\sigma$.
        In the former case the corresponding copies of $ut$ form a $p$-inversion, in the latter case the corresponding copies of $vt$ form a $p$-inversion.
        
        Thus we may assume that for any copy of any edge $e$ of $G$ and any vertex $t$ stacked on $e$, at least one corresponding vertex $t_i$ does not appear to the right of both of its parents in $\tilde{G}$.
        But then there is a subgraph in $\tilde{G}$ isomorphic to $G$ for which $\sigma$ is left-growing and whose restriction $w'$ of $\tilde{w}$ is $\varepsilon$-close to $w$.
        By assumption, this subgraph contains a $p$-inversion.
    \end{claimproof}

    Now we construct a sequence $\langle H_1,w_1 \rangle, \langle H_2,w_2 \rangle, \ldots$ of edge-weighted rooted $2$-trees, so that for every $k$, every left-growing vertex ordering of $\langle H_k,w_k \rangle$ contains a $k$-inversion.

    \begin{claim}\label{claim:inductive-problem}
        For each $k \geq 1$, any given edge $e^*$ of any given weight $w^*$, and any non-empty half-open interval $[x,y) \subset \mathbb{R}$, there exists a $2$-tree $H_k$ rooted at $e^*$ and an edge weighting $w_k$ of $H_k$ such that 
        \begin{itemize}
         \item $w_k(e^*) = w^*$ while $w_k(e) \in [x,y)$ for each $e \in E(H_k) - e^*$,
        \end{itemize}
        and in every left-growing vertex ordering $\sigma$ of $\langle H_k,w_k \rangle$ there is
        \begin{itemize}
            \item a $k$-inversion $(e_1,\ldots,e_k)$ without $e^*$, but whose right end is the right endpoint of $e^*$.
        \end{itemize}
    \end{claim}
    
    \begin{claimproof}
        We proceed by induction on $k$.
        For $k = 1$, it is enough to let $H_k$ be just the given root edge $e^*$ (with $w_k(e^*) = w^*$) and one vertex $t$ stacked onto $e^*$ with both edge-weights in $[x,y)$.
        In any left-growing vertex ordering, the edge between $t$ and the right endpoint of $e^*$ is the desired $1$-inversion.
        
        So for the remainder let us fix $k \geq 2$, and assume that the claim holds for $k-1$.
        Our task is to construct $\langle H_k,w_k \rangle$ for the given root edge $e^* = v'v''$, given weight $w^*$, and given interval $[x,y)$.
        By additive and multiplicative shifts of the weight function, we can assume for convenience that $[x,y) = [0,d+1)$ for $d = 2k^2$.
        Starting with $e^* = v'v''$ of weight $w_k(e^*) = w^*$, we construct $\langle H_k,w_k \rangle$ as follows.
        See \cref{fig:2-trees} for an illustration.
        \begin{itemize}
            \item Stack $d$ vertices $a_1,\ldots,a_d$ onto the root edge $e^*$, giving weight~$0$ to all $2d$ new edges.
            
            \item For each $i = 1,\ldots,d$
                \begin{itemize}
                    \item stack a vertex $b'_i$ onto edge $a_iv'$ giving weight~$0$ to both new edges,
                    
                    \item let $a_ib'_i$ be the root edge of $\langle H_{k-1},w_{k-1} \rangle$ with interval $[i,i+1)$, as well as the root of $\langle H_{k-1},w_{k-1} \rangle$ with interval $[d+1-i,d+1-i+1)$,
                    
                    \item stack a vertex $b''_i$ onto edge $a_iv''$ giving weight~$0$ to both new edges, and
                    
                    \item let $a_ib''_i$ be the root edge of $\langle H_{k-1},w_{k-1} \rangle$ with interval $[i,i+1)$, as well as the root of $\langle H_{k-1},w_{k-1} \rangle$ with interval $[d+1-i,d+1-i+1)$.
                \end{itemize}
        \end{itemize}
        
%        \begin{figure}[t]
%            \centering
%            \includegraphics{2-trees}
%            \caption{(Copy of \cref{fig:2-trees}.) Illustration of the $2$-tree $H_k$ with edge-weighting $w_k$, for $k\geq2$.}
%            \label{fig:2-trees-appendix}
%        \end{figure}

        Now consider any left-growing vertex ordering $\sigma$ of $H_k$.
        Since the construction is symmetric, we may assume without loss of generality that $v'$ is the right endpoint of the root edge $e^*$ in~$\sigma$.
        Hence, each of $a_1,\ldots,a_d,b'_1,\ldots,b'_d$ lies left of $v'$ since $\sigma$ is left-growing and $v'$ is a parent of all these vertices.
        By the pigeon-hole principle, there is a set $I_1 \subseteq [d]$ of at least $d/2 = k^2$ indices such that either $a_i \prec b'_i$ for all $i\in I_1$ or $b'_i \prec a_i$ for all $i \in I_1$.
        By symmetry\footnote{We are not using in the remainder that $a_iv''$ is an edge, while $b'_iv''$ is not.}, let us assume that $b'_i \prec a_i$ for all $i\in I_1$.
        
        By the Erd\H{o}s--Szekeres theorem, there is a subset $I_2 \subseteq I_1$ of at least $\sqrt{k^2} = k$ indices such that either $a_i \prec a_j$ for all $i,j \in I_2$ with $i < j$ or $a_j \prec a_i$ for all $i,j \in I_2$ with $i < j$.
        By symmetry, let us assume that $a_j \prec a_i$ for all $i,j \in I_2$ with $i < j$.\footnote{We use the symmetry of the two copies of $H_{k-1}$ with root edge $b'_ia_i$ here. In the assumed case, we focus on the copies of $H_{k-1}$ with intervals $[i,i+1)$. In the other case, we would consider the copies of $H_{k-1}$ with intervals $[d+1-i,d+1-i+1)$. We only need the interval weights to be increasing right-to-left.}
        Note that for the largest index $\hat{i} = \max(I_2)$ all vertices $a_i$ with $i \in I_2 - \hat{i}$ lie below the edge $a_{\hat{i}}v'$.
        
        Now consider for each $i \in I_2 - \hat{i}$ the copy of $2$-tree $H_{k-1}$ with root edge $b'_ia_i$ and interval $[i,i+1)$.
        By induction, for each $i \in I_2 - \hat{i}$, there is a $(k-1)$-inversion whose right end is $a_i$ and all of whose edge-weights lie in the interval $[i,i+1)$.
        
        Now, if for at least one such $i$, the left end of the corresponding $(k-1)$-inversion $(e_1,\ldots,e_{k-1})$ lies right of $a_{\hat{i}}$, then $(a_{\hat{i}}v',e_1,\ldots,e_{k-1})$ is the desired $k$-inversion, and we are done.
        Otherwise, each of the first edges from the $k-1$ corresponding $(k-1)$-inversions has its left endpoint to the left of $a_{\hat{i}}$.
        But then, taking edge $a_{\hat{i}}v'$ and these $k-1$ first edges in order of increasing weights again gives a desired $k$-inversion.
%        \todo{Reviewer~1: ``this sentence seems like a big shortcut, I would explain more why this works.''}
    \end{claimproof}

    To finish the proof of the theorem, it is now enough to take $\langle H_p,w_p \rangle$ from \cref{claim:inductive-problem} with any interval $[x,y)$, and then apply \cref{claim:wlog-left-growing} to obtain the desired edge-weighted $2$-tree $\langle G_p,w_p \rangle$ with $\pqn(G_p,w_p) \geq p$.
\end{proof}

\section{Omitted Proofs of \cref{sec:complexity}}

It remains to prove \cref{claim:pair-same-queues},
which is part of the proof of \cref{thm:npc-fixed-order}.

%\begin{figure}[t]
%	\centering
%	\includegraphics{nph-fixed-order}
%	\caption{(Copy of \cref{fig:nph-fixed-order}.)
%		Illustration of our reduction from coloring
%		circular-arc graphs (example instance on the left)
%		to the problem of determining,
%		for an edge-weighted graph with fixed vertex ordering,
%		the number of priority queues (resulting instance on the right).}
%	\label{fig:nph-fixed-order-appendix}
%\end{figure}

\PairSameQueues*
\label{claim:pair-same-queues*}
	
\begin{claimproof}
	Refer to \cref{fig:nph-fixed-order} for an illustration of our construction.
	We first show inductively that the edges on the left side
	are assigned to the same priority queues as their
	(adjacent) synchronization edges,
	and then we show the same for the edges on the right side.
	
	Let $a_1$ be the edge with the leftmost
	left endpoint corresponding to a vertex $a \in S$.
	The $k-1$ heavy edges below $a_1$ require $k-1$
	priority queues, none of which can be used by $a_1$
	or its synchronization edge since both form a forbidden
	pattern with each of the heavy edges.
	Therefore, $a_1$ and its synchronization edge
	are assigned to the only remaining priority queue.
	
	Let $b_1$ be the edge with the second leftmost
	left endpoint corresponding to a vertex $b \in S$.
	The $k-2$ heavy edges below $b_1$ require $k-2$
	priority queues, none of which can be used by $b_1$
	or its synchronization edge.
	Moreover, $a_1$ and its synchronization edge
	occupy another priority queue.
	Hence, only one priority queue is left
	for $b_1$ and its synchronization edge.
	This argument repeats for all subsequent vertices of~$S$.
	
	Now consider the edges on the right side belonging
	to vertices of~$S$.
	Here, the symmetric argument applies from right to left.
	Hence, the edges $a_2, b_2, \dots$ are assigned to
	the same priority queue as their synchronization edges,
	which in turn implies that each pair of edges corresponding to
	a vertex from~$S$ is assigned to the same priority queue.
\end{claimproof}

\end{document}